\documentclass[envcountsame]{llncs}

\usepackage[utf8]{inputenc}
\usepackage[T1]{fontenc}
\usepackage{amsmath}
\usepackage{amssymb}
\usepackage{xspace}
\usepackage{tikz}
\usetikzlibrary{backgrounds,positioning,fit,patterns,shadows,calc}

\renewcommand{\paragraph}[1]{\medskip\noindent{\bf #1.}\xspace}

\usepackage{fullpage}
\usepackage{booktabs}
\usepackage{threeparttable}

\makeatletter
\renewcommand{\paragraph}{%
  \@startsection{paragraph}{4}%
  {\z@}{1ex \@plus 1ex \@minus .2ex}{-1em}%
  {\normalfont\normalsize\bfseries}%
}
\makeatother

\usepackage{tikz}
\usetikzlibrary{backgrounds,positioning,fit,patterns,shadows,calc}

\usepackage{epsfig}
\usepackage{amsfonts}
\usepackage{amssymb}
\usepackage{amstext}
\usepackage{amsmath}

\usepackage{calligra}
\usepackage[T1]{fontenc}

\usepackage{thmtools}

\usepackage{paralist}

\usepackage{nameref}
\usepackage{colordvi}
\usepackage{color}
\definecolor{ForestGreen}{rgb}{0.1333,0.5451,0.1333}
\definecolor{DarkRed}{rgb}{0.8,0,0}
\definecolor{Red}{rgb}{1,0,0}
\usepackage[linktocpage=true,
	colorlinks,
	linkcolor=DarkRed,citecolor=ForestGreen,
	bookmarks=false,bookmarksopen,bookmarksnumbered]
	{hyperref}
\usepackage{cleveref}

\usepackage{thm-restate} 

\usepackage{xspace}
\usepackage{color}
\usepackage{enumitem}
\usepackage{comment}
\usepackage{caption}
\usepackage{subcaption}
\usepackage[noend]{algorithmic}
\usepackage[section,boxed]{algorithm}
\usepackage{algorithmicext}

\def\cH{\mathcal{H}}
\def\cA{\mathcal{A}}

\def\cC{\mathcal{C}}

\newcommand{\eps}{\varepsilon}
\renewcommand{\varepsilon}{\epsilon}
\newcommand{\st}{{\sf ST}\xspace}
\newcommand{\mcds}{{\sf MCDS}\xspace}
\newcommand{\scs}{{\sf SCS}\xspace}
\newcommand{\mis}{{\sf MIS}\xspace}

\newcommand{\rmds}{{\sf RMDS}\xspace}

\newcommand{\bmds}{{\sf BMDS}\xspace}

\newcommand{\dimension}{\textsc{dim}}

\newconstruct{\PROC}{\textbf{procedure}}{}{\ENDPROC}{\textbf{end on}}

\def\polylog{\operatorname{polylog}}

\newboolean{short}
\setboolean{short}{false}

\newcommand{\shortOnly}[1]{\ifthenelse{\boolean{short}}{#1}{}}
\newcommand{\onlyShort}[1]{\ifthenelse{\boolean{short}}{#1}{}}
\newcommand{\longOnly}[1]{\ifthenelse{\boolean{short}}{}{#1}}
\newcommand{\onlyLong}[1]{\ifthenelse{\boolean{short}}{}{#1}}

\makeatletter
\def\thmt@refnamewithcomma #1#2#3,#4,#5\@nil{%
  \@xa\def\csname\thmt@envname #1utorefname\endcsname{#3}%
  \ifcsname #2refname\endcsname
    \csname #2refname\expandafter\endcsname\expandafter{\thmt@envname}{#3}{#4}%
  \fi
}
\makeatother

\spnewtheorem{observation}{Observation}{\bfseries}{\itshape}

\ifdefined\ShowComment

\def\danupon#1{\marginpar{$\leftarrow$\fbox{D}}\footnote{$\Rightarrow$~{\sf #1 --Danupon}}}
\def\hsinhao#1{\marginpar{$\leftarrow$\fbox{H}}\footnote{$\Rightarrow$~{\sf #1 --Hsin-Hao}}}

\else

\def\danupon#1{}
\def\hsinhao#1{}

\fi

\title{Distributed Symmetry Breaking in Hypergraphs}

\titlerunning{Distributed Symmetry Breaking in Hypergraphs}

\author{
Shay Kutten\inst{1}\fnmsep\thanks{Research supported in part by the Israel Science Foundation and by the Technion TASP center.}
\and 
Danupon Nanongkai\inst{2}\fnmsep\thanks{This work was partially done while at ICERM, Brown University USA and Nanyang Technological University, Singapore.}
\and
Gopal Pandurangan\inst{3}\fnmsep\thanks{This work was done while at Nanyang Technological University and Brown University. Research supported in part by the following research grants: Nanyang Technological University grant M58110000, Singapore Ministry of Education (MOE) Academic Research Fund (AcRF) Tier 2 grant MOE2010-T2-2-082, Singapore MOE  AcRF Tier 1 grant MOE2012-T1-001-094, and a grant from the US-Israel Binational Science Foundation (BSF).}
\and
Peter~Robinson\inst{4}\thanks{Research supported by the grant Fault-tolerant Communication Complexity in Wireless Networks from the Singapore
MoE AcRF-2.}
}

\authorrunning{D. Nanongkai and H. Su}

\tocauthor{D. Nanongkai and H. Su}

\institute{Faculty of IE\&M, Technion, Haifa, Israel \and Faculty of Computer Science, University of Vienna, Austria \and Department of Computer Science, University of Houston, USA \and Department of Computer Science, National University of Singapore}

\begin{document}

\maketitle

\begin{abstract}

Fundamental local symmetry breaking  problems such as Maximal Independent Set (MIS) and coloring have been recognized as important by the community, and studied extensively in (standard) graphs. In particular, fast (i.e., logarithmic run time)   randomized algorithms are well-established for MIS and $\Delta +1$-coloring in both the LOCAL and CONGEST distributed computing models.  On the other hand, comparatively much less is known on the complexity of distributed symmetry breaking in {\em hypergraphs}.
In particular,  a key question is  whether a fast (randomized) algorithm for MIS exists  for hypergraphs.

In this paper, we study the distributed complexity of symmetry breaking in hypergraphs   by presenting distributed randomized algorithms for a variety of fundamental problems under a natural distributed computing model for hypergraphs.    We first show that MIS in hypergraphs (of arbitrary dimension) can be solved in $O(\log^2 n)$ rounds  ($n$ is the number of nodes of the hypergraph) in the LOCAL  model. 
We then present a key result of this paper ---  an $O(\Delta^{\eps}\polylog n)$-round   hypergraph MIS algorithm in the CONGEST model  where  $\Delta$ is the maximum node degree of the hypergraph and
$\eps > 0$ is any arbitrarily small constant. We also present distributed algorithms for coloring, maximal matching, and maximal clique in hypergraphs.

To demonstrate the usefulness 
of hypergraph MIS, we present applications of our hypergraph algorithm to solving problems in (standard) graphs.
In particular, the hypergraph MIS yields fast distributed algorithms for the  {\em balanced minimal dominating set} problem (left open in Harris et al. [ICALP 2013]) and the {\em minimal connected dominating set  problem}.  

Our work shows that while some local symmetry breaking problems such as coloring can be solved in polylogarithmic rounds in both the LOCAL and CONGEST models, for many other hypergraph problems  such as MIS, hitting set, and maximal clique, it remains challenging to obtain polylogarithmic time algorithms in the CONGEST model.  This work is a step towards understanding this dichotomy in the complexity of hypergraph problems as well
as using hypergraphs to design fast distributed algorithms for problems in (standard) graphs.

\end{abstract}

\tikzstyle{v}=[circle,draw=black,fill=white!30,thick,inner sep=2pt,minimum size=7mm,circular drop shadow]
\tikzstyle{b}=[v,double]

\onlyShort{\vspace{-0.4in}}
\section{Introduction} \label{sec:intro}
\onlyShort{\vspace{-0.1in}}
The importance, as well as the difficulty, of solving problems on hypergraphs was pointed out recently by Linial, in his Dijkstra award talk \cite{linial-talk}.
While standard graphs\footnote{Henceforth,
when we say a graph, we just mean a standard (simple) graph.} model {\em pairwise} interactions well,
 hypergraphs  can be used to model {\em multi-way} interactions. For example, social network interactions include several individuals as a group, biological interactions involve several entities (e.g., proteins) interacting at the same time, distributed systems can involve several agents working together, or multiple clients who share a server (e.g., a cellular base station), or multiple servers who share a client, or shared channels in a wireless network. In particular, hypergraphs are especially useful in modelling social networks (e.g., \cite{wasserman}) and
 wireless networks (e.g., \cite{avin}).
Unfortunately, as pointed out by Linial, much less is known for hypergraphs than for graphs.
The focus of this paper is studying the complexity of fundamental local
symmetry breaking
problems in {\em hypergraphs}\footnote{Formally, a hypergraph $(V,F)$
consists of a set of (hyper)nodes $V$ and a collection $F$ of subsets of $V$; the sets that belong to $F$ are called  {\em hyperedges}.
The {\em dimension} of a hypergraph is the maximum number of hypernodes that belong to a hyperedge. Throughout, we will use $n$ for the number of nodes,
$m$ for the number of hyperedges,  and $\Delta$ for the degree of the hypergraph which is the maximum node degree (i.e., the maximum number of edges a node is in).
 A standard graph is a  hypergraph of dimension 2.}. A related goal is to utilize these hypergraph algorithms for solving (standard) graph problems.

%
In the area of distributed computing for (standard) graphs,
fundamental local symmetry breaking problems such as Maximal Independent Set (MIS) and coloring have been studied extensively (see e.g., \cite{Luby86,Linial92,elkin-book,Pel00,kuhn-local} and the references therein).
Problems such as MIS and coloring  are ``local'' in the sense
that a solution can be {\em verified} easily by purely local  means (e.g.,  each node  communicating only with its neighbors),
but the  solution itself  should satisfy a  global property
(e.g., in the case of coloring, every node in the graph should have a color different from its neighbors and the total number of colors is at most $\Delta + 1$, where $\Delta$ is the maximum node degree).
Computing an MIS or coloring locally is non-trivial because of the difficulty of {\em symmetry breaking}: nodes have to decide on their choices  (e.g., whether they belong to the MIS or not) by only looking at a {\em small} neighbourhood around it. (In particular,
to get an algorithm running in $k$ rounds,
each node $v$ has to make its decision by looking only
at information on nodes within
distance $k$ from it.)
Some of the most celebrated results in distributed algorithms are  such fast localized algorithms.
In particular, $O(\log n)$-round (randomized) distributed algorithms are well-known for MIS \cite{Luby86} and $\Delta +1$-coloring \cite{elkin-book} in both the LOCAL and CONGEST distributed computing models \cite{Pel00}. 
%

Besides the interest in understanding the complexity of  fundamental problems,
the solutions to such localizable symmetry breaking problems had many obvious applications. Examples are scheduling (such as avoiding the collision of radio transmissions, see e.g. \cite{ephremides},
 \cite{chlamtac-kutten},
 or matching nodes such that each pair can communicate in parallel to the other pairs, see e.g. \cite{d2matching}), resource management (such as assigning clients to servers, see, e.g. \cite{azar-naor-rom}),
and even for obtaining $O(Diameter)$ solutions to global problems that cannot be solved locally, such as MST computation \cite{DBLP:journals/siamcomp/GarayKP98,KDOM}.

In contrast to graphs which have been extensively studied in the context of distributed algorithms,
many problems become much more challenging in the context of hypergraphs. An outstanding example is the  MIS problem, whose local solutions for graphs were mentioned above.
On the other hand, in hypergraphs (of arbitrary dimension) the complexity of MIS  is wide open. (In a hypergraph, an MIS is a maximal subset $I$ of hypernodes such that no subset of $I$ forms an hyperedge.)
Indeed, determining the parallel complexity (in the PRAM model) of the Maximal Independent Set (MIS) problem in hypergraphs (for arbitrary dimension) remains as one of the most important open problems in parallel computation; in particular, a key open problem is whether there exists a
polylogarithmic time PRAM algorithm \cite{karp-ram,BeameL90,Kelsen92}.   As discussed later, efficient CONGEST model distributed algorithms 
that uses simple local computations will also give efficient PRAM algorithms. 


\onlyShort{\vspace{-0.15in}}
\subsection{Main Results}
\onlyShort{\vspace{-0.1in}}


We present distributed (randomized) algorithms for a variety of fundamental problems under a natural distributed computing model for hypergraphs (cf. Section \ref{sec:prelim}). 

\paragraph	{Hypergraph MIS.}
A main focus is the hypergraph MIS problem which has been the subject of extensive research in
the PRAM model  (see e.g., \cite{karp-ram,KarpUW88,Kelsen92,BeameL90,Luczak}).  We first show that MIS in hypergraphs (of arbitrary dimension) can be solved in  $O(\log^2 n)$ distributed rounds ($n$ is the number of nodes of the hypergraph) in the LOCAL  model (cf. Theorem \ref{thm:mis}).
We then present  an $O(\Delta^{\eps} \polylog n)$ round algorithm for finding a MIS in hypergraphs of arbitrary dimension in the CONGEST model,  where $\Delta$ is the maximum degree of the hypergraph 
(we refer to Theorem \ref{thm:mis} for a precise statement of the bound) and $\eps > 0$ is any small positive constant. 
In the distributed computing model (both LOCAL and CONGEST), computation within a node is free; in one round, each node is allowed to compute any function of its current data. However, in our CONGEST model algorithms,  each processor will perform very simple computations (but this is not true in the LOCAL model). In particular, each step of any node $v$ can be simulated in $O(d_v)$  time by a single processor or in $O(\log m)$  time with $d_v$ processors.  Here, $d_v$ is the degree of
the node in the {\em server-client} computation model --- cf. Section \ref{sec:prelim};   $d_v = O(m)$, where $m$ is the number of hyperedges.  From these remarks, it follows that our algorithms can be simulated on the PRAM model to within an $O(\log m)$ factor slowdown using $O(m+n)$ processors.
Thus our CONGEST model algorithm also implies a  PRAM algorithm for hypergraph MIS running  in $O(\Delta^{\eps}\polylog n \log m)$ rounds using a linear number of processors for a hypergraph of arbitrary dimension. 



\paragraph{Algorithms for standard graph problems using hypergraph MIS.}
In addition to the importance of hypergraph MIS as a hypergraph problem, we outline its importance to solving several natural symmetry breaking problems in (standard) graphs too. For  the results discussed below, we assume the CONGEST model.

Consider first the following  graph  problem
called  the {\em restricted minimal dominating set (RMDS)} problem which arises as a key subproblem in other problems
that we discuss later. We are given a (standard) graph $G = (V,E)$ and a subset of nodes $R \subseteq V$, such
that  $R$ forms a dominating set in $G$ (i.e., every node $v \in V$ is either adjacent to $R$ or belongs to $R$).
It is required to find a {\em minimal} dominating set {\em in $R$} that dominates $V$. 
(The minimality means that no subset of the solution can dominate $V$; it is easy to verify the minimality condition locally.)
Note that if $R$ is $V$ itself, the problem can be solved by finding a MIS of $G$, since a MIS is also a minimal dominating set (MDS); hence an $O(\log n)$ algorithm exists. However, if $R$ is some arbitrary proper subset of $V$ 
(such that $R$ dominates $V$), then
no  distributed algorithm running even in sublinear (in $n$) time (let alone polylogarithmic time) is known.
Using our hypergraph MIS algorithm, we design a  distributed algorithm for RMDS  running  in $O(\min\{\Delta^{\eps}\polylog n, n^{o(1)}\})$ rounds in the CONGEST model ($\Delta$ is the maximum node degree of the graph) --- cf., Section \ref{sec:rmds}.

RMDS arises naturally as the key subproblem in the solution of other  problems, in particular, the {\em balanced minimal dominating set (BMDS)} problem \cite{balanced-minimal}
and the {\em minimal connected dominating set (MCDS)} problem. Given a (standard) graph, the BMDS problem (defined formally in Section \ref{sec:bmds})  asks for a minimal dominating set whose average degree is small with respect to the average degree of the graph; this has applications to load balancing and fault-tolerance \cite{balanced-minimal}.  It was shown that such a set exists and can be found using a {\em centralized} algorithm \cite{balanced-minimal}. Finding a fast distributed algorithm  was a key problem left open in \cite{balanced-minimal}. In Section \ref{sec:mcds}, we use our hypergraph MIS algorithm of Section~\ref{sec:hyper} to present an  $\tilde{O}(D+ \min\{\Delta^{\eps}, n^{o(1)}\})$ round algorithm (the notation $\tilde{O}$ hides a $\polylog n$ factor) for BMDS problem (in the CONGEST model), where
$D$ is the diameter (of the input standard graph) and $\Delta$ is the maximum node degree. 

The MCDS problem is a variant (similar to variants studied in the context of wireless networks, e.g. \cite{localized-cds}) of the well-studied
{\em minimum} connected dominating set problem (which is NP-hard) \cite{approx-min-connected,routing-min-connected}. 
  In the MCDS problem, we require a dominating set that is connected and is {\em minimal} (i.e.,
no subset of the solution is a MCDS). In contrast to the approximate minimum connected dominating set problem
(i.e., finding a connected dominating set that is not too large compared to the optimal) which admits efficient distributed algorithms \cite{Dubhashi,ghaffari} (polylogarithmic run time algorithms are known
for both the LOCAL and CONGEST model for the unweighted case), we show that it is impossible to obtain an efficient distributed algorithm for MCDS. 
 In Section \ref{sec:mcds}, we use our hypergraph MIS algorithm of Section \ref{sec:hyper}  as a subroutine to construct a distributed algorithm for MCDS that runs in time
$\tilde O(D (D\min\{\Delta^{\eps}, n^{o(1)}\} +\sqrt{n}) )$. We also show that $\tilde \Omega(D + \sqrt{n})$ is a lower bound
on the run time for any distributed  MCDS algorithm.

\paragraph{Algorithms for other hypergraph problems.}
Besides MIS (and the above related standard graph problems), we also study distributed algorithms for coloring, maximal matching,   and maximal clique in hypergraphs\onlyLong{.}\onlyShort{ in the full paper.}  We show that a $\Delta+1$-coloring of a hypergraph (of any arbitrary dimension) can be computed in $O(\log n)$ rounds (this generalizes the result for standard graphs).  We also show that maximal matching in hypergraphs can be solved in $O(\log m)$ rounds.
Maximal clique is a less-studied problem, even in the case of graphs, but nevertheless interesting.
Given a (standard) graph $G=(V,E)$,  a maximal clique (MC) $L$ is subset of $V$ such that $L$ is a clique in $G$
and is maximal (i.e., it is not contained in a bigger clique). MC is related to MIS since any MIS in the complement graph $G^c$ is an MC in $G$.
For a hypergraph, one can define an MC with respect to the server graph (cf. Section \ref{sec:prelim}). 
Finding MC has applications in finding a {\em non-dominated coterie} in quorum systems \cite{makino}.
We show that an MC in a hypergraph can be found in  $O(\dimension \log n)$ rounds, where $\dimension$ is the dimension of the hypergraph and $n$ is the number of nodes. All the above results hold in the CONGEST model as well. 

\onlyShort{\vspace{-0.15in}}
\subsection{Technical Overview and Other Related Work}
\onlyShort{\vspace{-0.1in}}
We study two natural network models for computing with hypergraphs --- the {\em server-client} model
and the {\em vertex-centric} models (cf. Section \ref{sec:prelim}). The server-client model is commonly used in packing and covering problems such as set cover and packing LPs (e.g., \cite{Suomela13,AstrandS10,PapadimitriouY93,KuhnMW06,BartalBR97,Kuhn2005-thesis}).  It is also a natural model for the facility location problem (e.g., \cite{MoscibrodaW05,PanditP09}).  The vertex-centric model was considered in, e.g., \cite{KoufogiannakisY11}. 
%
%
%
Our algorithmic results apply to both models (except the one on maximal matching). 

The distributed MIS problem on hypergraphs is significantly more challenging than that on (standard) graphs.
Simple variants/modifications of the distributed algorithms on graphs (e.g.,  Luby's algorithm and its variants \cite{Luby86,MetivierRSZ11,Pel00})  do not seem to work for higher dimensions, even
for hypergraphs of dimension 3. For example, running Luby's algorithm or its permutation variant \cite{Luby86} on a (standard) graph by replacing each hyperedge with a clique does not work --- in the graph there can be only one
node in the MIS, whereas in the hypergraph all nodes of the clique, except one, can be in the MIS. 
It has been conjectured by Beame and Luby \cite{BeameL90} that a generalisation of the permutation variant of an algorithm due to Luby \cite{Luby86} can give
a $\polylog(m+n)$ run time in the PRAM model, but this has not been proven so far (note that this bound itself can be large, since $m$ can be exponential in $n$). 
%

\danupon{Papers that are very related to us but I don't know where to mention it is \cite{AstrandS10} where they approximate set cover. They also cite a few other papers related to this problem.}


Our distributed hypergraph MIS algorithm (Section \ref{sec:hyper}) consists of several ingredients. A key ingredient is the {\em decomposition lemma} (cf. Lemma \ref{thm:decomposition congest}) that shows that the problem can be reduced to solving a MIS problem in a low diameter network.  The lemma is essentially an application of the {\em network decomposition} algorithm of Linial and Saks \cite{LinialS93}.  This applies to the CONGEST model as well --- the main task in the proof is to show that the Linial-Saks decomposition works for (both) the hypergraph models in the CONGEST setting. The polylogarithmic run time bound for the LOCAL model follows easily 
from the decomposition lemma. However, this approach fails in the CONGEST model, since it involves collecting a lot of information at some nodes. The next ingredient is to show how the PRAM algorithm of Beame and Luby \cite{BeameL90} can be simulated efficiently in the distributed setting; this we show is
possible in a low diameter graph.  Kelsen's analysis \cite{Kelsen92} of Beame-Luby's algorithm (which shows a polylogarithmic time bound  in the PRAM model for {\em constant} dimension hypergraphs) immediately gives a polylogarithmic round algorithm in the CONGEST
model for a hypergraph of {\em constant} dimension.  To obtain the $\tilde{O}(\Delta^{\eps})$ algorithm  (for any constant $\eps > 0$) for 
a hypergraph of {\em arbitrary} dimension in the CONGEST model, we use another 
ingredient: we generalize a theorem of Turan  (cf. Theorem \ref{thm:Turan}) for hypergraphs --- this shows that a hypergraph of low average degree has a  {\em large} independent set. We show further that such a large independent set can be found  when the network diameter is $O(\log n)$.  Combining this theorem
with the analysis of Beame and Luby's algorithm gives the result for the CONGEST model for any dimension. Our CONGEST model algorithm, as pointed out earlier, also implies a $\tilde{O}(\Delta^{\eps})$ round algorithm for the PRAM model.
 Recently, independently of our result,
 Bercea et al.\cite{aravind2} use a similar approach to obtain an improved algorithm for the PRAM model. In particular, they improve Kelsen's analysis of Beame-Luby algorithm to apply also for slightly {\em super-constant} dimension. This improved analysis of Kelsen also helps us
 in obtaining a slightly better bound (cf. Theorem \ref{thm:mis}).
 
We apply our hypergraph  MIS algorithm  to solve two key problems --- BMDS and MCDS. 
The BMDS problem was posed in Harris et al. \cite{balanced-minimal}, but no efficient distributed algorithm was known.
A key bottleneck was solving the RMDS problem which appears as a subroutine in solving BMDS.
In the current paper, we circumvent this bottleneck by treating the RMDS problem as a problem on hypergraphs.

The MCDS problem, to the best of our knowledge, has not been considered before and seems significantly
harder to solve in the distributed setting compared to the more well-studied approximate version of the connected dominating set problem \cite{Dubhashi,ghaffari}. The key difficulty is being {\em minimal} with respect to {\em both}
connectivity and domination. We use a layered approach to the problem, by first constructing a breadth-first tree (BFS) and then adding nodes to the MCDS, level by level of the tree (starting with the leaves).
\onlyLong{
We make sure that nodes added to the MCDS in level $i$ dominates the nodes in level $i+1$ and is also minimal. To be minimal with respect to connectivity we cluster nodes that are in MCDS at level $i+1$ by connected components and treat these as super-nodes. To minimally dominate these super nodes we use the hypergraph MIS algorithm; however there is a technical difficulty of simulating
the hypergraph algorithm on super-nodes. We show that such a simulation can be done efficiently by reducing the dimension of the constructed hypergraph (\onlyLong{cf. Lemma \ref{lem:logn}}\onlyShort{cf. Lemma 4.4 in the full paper in Appendix}) which show that hypergraph MIS on a hypergraph of arbitrary dimension can be reduced to solving a equivalent problem in 
a hypergraph of $\polylog(m+n)$ dimension with only $O(\log n)$ factor slow down. 
}
%
%
We also show a lower bound of $\tilde{\Omega}(D+\sqrt{n})$ for the MCDS problem  by using the techniques
of Das Sarma et al. \cite{STOC11}. This lower bound holds even when $D=\polylog n$. In this case, our upper bound is tight up to a $\polylog n$ factor.
We also show that $\Omega(D)$ is a {\em universal} lower bound for MCDS as well as for maximal clique and spanning tree problems, i.e., it applies essentially to all graphs. \onlyShort{These are shown in the full paper.}

\onlyShort{\vspace{-0.1in}}
\section{Preliminaries}
\label{sec:prelim}
\onlyShort{\vspace{-0.1in}}
A hypergraph $\cH$ consists of a set $V(\cH)$ of $n$  (hyper)nodes and a set family $E(\cH)$ of $m$ hyperedges, each of which  is a subset of $V(\cH)$. 
We define the \emph{degree of node $u$} to be the total number of hyperedges that $u$ is contained in. 
Furthermore, we define the \emph{degree of the hypergraph}, denoted by $\Delta$, as the maximum over all hypernode degrees.
The size of each hyperedge is bounded by the \emph{dimension $\dimension$} of $\cH$; note that a hypergraph of dimension $2$ is a graph. 

We now introduce our main model of computation.
In our distributed model, $\cH$ is realized as a (standard) undirected bipartite graph $G$ with vertex sets $S$ and $C$ where $|S|=n$ and $|C|=m$.
We call $S$ the set of \emph{servers} and $C$ the set of \emph{clients} and denote this realization of a hypergraph as the \emph{server-client model}.
That is, every vertex in $S$ corresponds to a vertex in $\cH$ and every vertex in $C$ corresponds to a hyperedge of $\cH$.
For simplicity, we use the same identifiers for vertices in $C$ as for the hyperedges in $\cH$.
There exists a ($2$-dimensional) edge in $G$ from a server $u \in S$ to a client $e \in C$ if and only if $u \in e$. 
See \Cref{fig:hyper} for an example.
Thus, the degree of $\cH$ is precisely the maximum degree of the servers and the dimension of $\cH$ is given by the maximum degree of the clients.
\begin{figure*}
\subcaptionbox{\label{fig:hyper}}[0.25\linewidth]{
  \begin{tikzpicture}[scale=0.7,every node/.style={scale=0.7}]
  \small
  \node[v] (u1) at (0,0) {$u_1$};
  \node[v,below right of=u1,xshift=0.5cm,yshift=-0.5cm] (u2) {$u_2$};
  \node[v,below left of=u2,xshift=-0.5cm,yshift=-0.5cm] (u3) {$u_3$};
  \node[v,below right of=u3,xshift=0.5cm,yshift=-0.5cm] (u4) {$u_4$};

  \begin{pgfonlayer}{background}
  \begin{scope}[fill opacity=0.8]
    \filldraw[fill=yellow!70] ($(u1)+(-0.5,0)$) 
        to[out=90,in=180] ($(u1) + (0,0.5)$) 
        to[out=0,in=90] ($(u2) + (0.5,0)$) 
        to[out=270,in=0] ($(u3) + (0,-0.5)$)
        to[out=180,in=270] ($(u3) + (-0.5,0)$)
        to[out=90,in=225] ($(u2)+(-0.7,-0.2)$)
        to[out=45,in=270] ($(u1)+(-0.5,0)$);

    \filldraw[fill=green!70,thick,dashed] ($(u2)+(-0.5,0.2)$)
        to[out=90,in=90] ($(u2)+(0.5,0.2)$)
        to[out=270,in=90] ($(u4)+(0.5,-0.2)$)
        to[out=270,in=270] ($(u4)+(-0.5,-0.2)$)
        to[out=90,in=270] ($(u2)+(-0.5,0.2)$);

    \filldraw[fill=red!70,thick,dotted] ($(u3)+(-0.5,0.2)$)
        to[out=115,in=135] ($(u3)+(0.3,0.5)$)
        to[out=315,in=115] ($(u4)+(0.5,0)$)
        to[out=315,in=315] ($(u4)+(-0.3,-0.4)$)
        to[out=135,in=285] ($(u3)+(-0.5,0.2)$);
    \end{scope}
  \end{pgfonlayer}
\end{tikzpicture}
}
\subcaptionbox{\label{fig:bipartite}}[0.3\linewidth]{
\begin{tikzpicture}[scale=0.5,every node/.style={scale=0.6}]
\tikzstyle{link}=[-,black,thick,auto]
  \small
  \node (servers) {servers\phantom{l}};
  \node[xshift=0.3cm,right of=servers] (clients) {clients};
  \node[v,below of=servers] (u1) {$u_1$};
  \node[v,below of=u1] (u2) {$u_2$};
  \node[v,below of=u2] (u3) {$u_3$};
  \node[v,below of=u3] (u4) {$u_4$};

  \draw[gray,dashed] ($(servers.north)+(0.65,0)$) -- ($(u4.south)+(0.65,0)$);

  \node[v,below right of=u1,xshift=0.5cm,fill=yellow!70] (e1) {$e_1$};
  \node[v,below right of=u2,xshift=0.5cm,fill=green!70,dashed] (e2) {$e_2$};
  \node[v,below right of=u3,xshift=0.5cm,fill=red!70,dotted] (e3) {$e_3$};

  \draw[link] (u1) to (e1);
  \draw[link] (u2) to (e1);
  \draw[link] (u3) to (e1);
  \draw[link] (u2) to (e2);
  \draw[link] (u4) to (e2);
  \draw[link] (u3) to (e3);
  \draw[link] (u4) to (e3);
\end{tikzpicture}
}
\subcaptionbox{\label{fig:server}
}[0.2\linewidth]{
\begin{tikzpicture}[scale=0.5,every node/.style={scale=0.6}]
\tikzstyle{link}=[-,black,thick,auto]
  \small
  \node[v] (u1) at (0,0) {$u_1$};
  \node[v,right of=u1,xshift=0.5cm] (u2) {$u_2$};
  \node[v,below of=u1,yshift=-0.5cm] (u3) {$u_3$};
  \node[v,below of=u2,yshift=-0.5cm] (u4) {$u_4$};

  \draw[link] (u1) to (u2);
  \draw[link] (u1) to (u3);
  \draw[link] (u2) to (u3);
  \draw[link] (u2) to (u3);
  \draw[link] (u2) to (u4);
  \draw[link] (u4) to (u3);
\end{tikzpicture}
}
\caption{\small Figure~(\subref{fig:hyper}) depicts a hypergraph consisting of vertices $u_1,\dots,u_4$ and edges $e_1=\{u_1,u_2,u_3\}$, $e_2=\{u_2,u_4\}$, and $e_3=\{u_3,u_4\}$. Figures~(\subref{fig:bipartite}) and (\subref{fig:server}) respectively show this hypergraph in the bipartite server-client model and the vertex-centric model.}
\end{figure*}
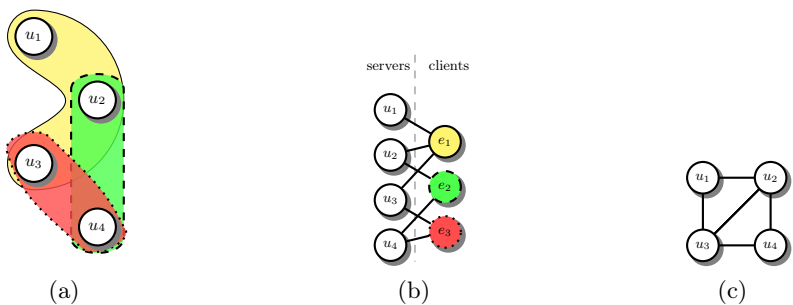

An alternative way to model a hypergraph $\cH$ as a  distributed network is the \emph{vertex-centric} model (cf.\ \Cref{fig:server}). Here, the nodes are exactly the nodes of $\cH$ and there exists a communication link between nodes $u$ and $v$ if and only if there exists a hyperedge $e \in E(\cH)$ such that $u, v \in e$.
Note that in this model, we assume that every node locally knows all hyperedges in which it is contained.
For any hypergraph $\cH$, we call the above underlying communication graph in  the vertex-centric model (which is a standard graph) the {\em server graph}, denoted by $G(\cH)$. 


We consider the standard synchronous round model (cf.\ \cite{Pel00}) of communication.
That is, each node has a unique id (arbitrarily assigned from some set of size polynomial in $n$) and executes an instance of a distributed algorithm that advances in discrete {\em rounds}.
To correctly model the computation in a hypergraph, we assume that each node knows whether it is a server or a client.
In each round every node can communicate with its neighbors (according to the edges in the server-client graph) and perform some local computation.
We do not assume shared memory and nodes do not have any a priori knowledge about the network at large.

We will consider two types of  models --- CONGEST and LOCAL \cite{Pel00}. In the CONGEST model, only a $O(\log n)$-sized
message can be sent across a communication edge per round. In the LOCAL model, there is no such restriction.
Unless otherwise stated, we use the CONGEST model in our algorithms.

\onlyShort{Due to lack of space, the complete proofs can be found in the full paper.}
\onlyShort{\vspace{-0.25cm}}
\section{Distributed Algorithms for Hypergraph MIS Problem}
\onlyShort{\vspace{-0.2cm}}
\label{sec:hyper}
We present randomized distributed algorithms and
prove the following for the hypergraph MIS problem:

\begin{theorem} \label{thm:mis}
The hypergraph MIS problem can be solved  in the following expected time\footnote{Our time bounds can also be easily shown to hold with high probability, i.e., with probability $1 -1/n$.} in both vertex-centric and server-client representations.  
\begin{compactenum}
	\item $O(\log^2 n)$ time in the LOCAL model.
	\item  $O(\log^{(d+4)!+4} n)$ time\footnote{As is common, we use the notation $\log^f n$ which is the same as $(\log n)^f$.} in the CONGEST model when the input hypergraph has constant dimension $d$. 
	\item $O(\min\{\Delta^{\epsilon}\log^{(1/\epsilon)^{O(1/\epsilon)}} n, \sqrt{n}\})$ time in the CONGEST model for any dimension, where $\epsilon$ is such that $1\geq \epsilon \geq \frac{1}{\frac{\log\log n}{c\log\log\log n}-1}$ from some (large) constant $c$. (In particular, $\Delta^{\epsilon}\log^{(1/\epsilon)^{O(1/\epsilon)}} n$ becomes $\Delta^{o(1)}n^{o(1)}$ when we use $\epsilon =  \frac{1}{\frac{\log\log n}{c\log\log\log n}-1}$.)
\end{compactenum}
\end{theorem}
In Section \ref{sec:decomposition}, we prove a {\em decomposition lemma} which plays an important role in achieving all the above results. 

\onlyShort{\vspace{-0.2cm}}
\subsection{Low-Diameter Decomposition}\label{sec:decomposition}
\onlyShort{\vspace{-0.2cm}}

First, we note that, for solving MIS, it is sufficient to construct an algorithm that solves the following {\em subgraph-MIS} problem on low-diameter networks.

\begin{definition}[Subgraph-MIS Problem]
In the Subgraph-MIS problem, we are given an $n$-node network $G$. This network is either in a vertex-centric or server-client representation of some hypergraph $\cH$. Additionally, we are given a subnetwork $G'$ of $G$ representing a sub-hypergraph\footnote{Given a subset $V' \subseteq V$,  a sub-hypergraph of $\cH$ is simply a hypergraph induced by $V'$ --- except hyperedges that contain vertices that do not belong to $V'$, all other hyperedges of $\cH$ (which intersect with $V'$) are present in the sub-hypergraph.}   $\cH'$ of $\cH$. The goal is to find an MIS of $\cH'$. 
\end{definition}

\begin{lemma}[Decomposition Lemma] \label{thm:decomposition congest}
For any function $T$, if there is an algorithm $\cA$ that solves subgraph-MIS on CONGEST server-client (respectively vertex-centric) networks $G$ of $O(\log n)$ diameter in $T(n)$ time (where $n$ is the number of nodes in $G$), then there is an algorithm for MIS on CONGEST server-client (respectively vertex-centric) networks of {\em any} diameter that takes $O(T(n)\log^4 n)$ time. 
\end{lemma}

The main idea of the lemma is to run the {\em network decomposition} algorithm of Linial and Saks \cite{LinialS93} and simulate $\cA$ on each cluster resulting from the decomposition. The only part that we have to be careful is that running $\cA$ simultaneously on many clusters could cause a congestion. We show that this can be avoided by a careful scheduling. The details are as follows.

%

The network decomposition algorithm of \cite{LinialS93} produces an {\em $O(\log n)$-decomposition with weak-diameter $O(\log n)$}. That is, given a (two-dimensional) graph $G$, it partitions nodes into sets $S_1, S_2, \ldots S_k$ and assigns color $c_i\in \{1, 2, \ldots, O(\log n)\}$ to each set $S_i$ with the following properties: 
\begin{compactitem}
\item the distance between any two nodes in the same set $S_i$ is $O(\log n)$, and 
\item any two neighboring nodes of the same color must be in the same set (in other words, any two ``neighboring'' sets must be assigned different colors).
\end{compactitem}

This algorithm takes $O(\log^2 n)$ time even in the CONGEST model~\cite{LinialS93}. We use the above decomposition algorithm to decompose the server graph $G(\cH)$ (cf. \Cref{sec:prelim}) of the input hypergraph. The result is the partition of hypernodes (servers) into colored sets satisfying the above conditions (in particular, two nodes sharing the same hyperedge must be in the same partition or have differnet colors). In addition, we modify the Linial-Saks (LS) algorithm to produce low-diameter  subgraphs that contain these sets with the property that subgraphs of the same color have ``small overlap''.

\begin{lemma}\label{claim:decomposition}
Let $G$ be the input network (server-client or vertex-centric model) representing hypergraph $\cH$. In $O(\log^3 n)$ time and for some integer $k$, we can partition hypernodes into $k$ sets $S_1, \ldots, S_k$, produce $k$ subgraphs of $G$ denoted by $G_1, G_2, \ldots G_k$, and assign color $c_i\in \{1, 2, \ldots, O(\log n)\}$ to each subgraph $G_i$, with the following properties: 
\begin{compactenum}
\item For all $i$, $G_i$ has diameter $O(\log n)$ and $S_i\subseteq V(G_i)$. 
\item For any $S_i$ and $S_j$ that are assigned the same color (i.e. $c_i=c_j$), there is no hyperedge in $\cH$ that contains hypernodes (servers) in both $S_i$ and $S_j$. 
\item Every edge in $G$ is contained in $O(\log^3 n)$ graphs $G_{i_1}, G_{i_2}, \ldots$ 
\end{compactenum}
\end{lemma}
Observe that the first two properties in \Cref{claim:decomposition} are similar to the guarantees of the Linial-Saks algorithm, except that \Cref{claim:decomposition} explicitly gives low-diameter graphs that contain the sets $S_1, \ldots, S_k$. The third property guarantees that such graphs have ``small congestion''. 
\onlyLong{
\begin{proof}
Note that the Linial-Saks algorithm works as follows. The algorithm runs in iterations where in the $i^{th}$ iteration it will output sets of color $i$. In the $i^{th}$ iteration, each vertex $y$ selects an integer radius $r_y\in \{1, \ldots, O(\log n)\}$ at random (according to some distribution). Then it broadcasts its ID and the value $r_y$ to all nodes within distance $r_y$ of it. For every node $v$, after receiving all such messages from other nodes, selects the node with highest ID from among nodes $y$ that sends their IDs to $v$; denote such node by $C(v)$. For any node $y$, define set $S_y$ as the set that contains every node $v$ that has $C(v)=y$ and its distance to $y$ is {\em strictly} less than $r_{y}$. We call $S_y$ the {\em set centered at $y$} (note that $y$ might not be in $S_y$). 
All sets in this iteration receives color $i$. The distance between every pair of nodes $u$ and $v$ in any set $S_y$ is $O(\log n)$ since their distance to $y$ is $O(\log n)$. We can guarantee that there are no two neighboring nodes $u$ and $v$ in different sets because otherwise $C(u)=C(v)$ (this crucially uses the fact that sets are formed by nodes $v$ whose distance to $C(v)$ is strictly less than $r_{C(v)}$). By carefully picking the distribution of $r_y$, \cite{LinialS93} shows that the number of iterations is $O(\log n)$. 

\smallskip
The following is one simple (although not the most efficient) way to simulate the above algorithm in the server-client CONGEST model to compute $S_1, \ldots, S_k$. We implement each iteration of the above algorithm in {\em sub-iterations}. In the beginning of the $j^{th}$ sub-iteration, every server $y$ with $r_y=j$ sends its ID to its neighboring clients. We then repeat the following for $2j-1$ steps: every node (client or server) sends the maximum ID that it receives to its neighbors. It is easy to see that after all sub-iterations every server $v$ receives the maximum ID among the IDs of servers $y$ such that $r_y=j$ and the distance between $y$ and $v$ in the server graph is at most $j$. Since $r_y=O(\log n)$ for every $y$, there are $O(\log n)$ sub-iterations and each sub-iteration takes $O(\log n)$ time. After all sub-iterations, every server $v$ can select $C(v)$. Thus, we can simulate the Linial-Saks algorithm in $O(\log^3 n)$ time. (Simulating Linial-Saks algorithm on the vertex-centric model can be done similarly except that we will have  $j-1$ sub-iterations instead of $2j-1$.)

We now construct $G_1, \ldots, G_k$. At any sub-iteration above, if a node $v$ sends the ID of some node $y$ to its neighbors, we add its neighbors and all edges incident to $v$ to $G_y$ (corresponding to set $S_y$). Clearly, $S_y$ is contained in $V(G_y)$ since $G_y$ contains all nodes that receive the ID of $y$. This process also guarantees that $G_y$ has $O(\log n)$ diameter since every node in $G_y$ can reach $G_y$ in $O(\log n)$ hops by following the path that the ID of $y$ was sent to it. Additionally, since the simulation of the Linial-Saks algorithm finishes in $O(\log^3 n)$ rounds, and in each round we add an edge $(u, v)$ to at most two subgraphs, we have that every edge is in $O(\log^3 n)$ subgraphs.
\end{proof}

\begin{proof}[Proof of \Cref{thm:decomposition congest}]
We decompose the network as in \Cref{claim:decomposition}. Then, we use $\cA$ to compute MIS iteratively in $O(\log n)$ iterations as follows. At the $i^{th}$ iteration, we consider each set $S_t$ and graph $G_t$ of color $i$. We will decide whether each node in $S_t$ will be in the final solution of MIS or not.  We assume that we already did so for sets of colors $1, 2, \ldots, i-1$. 

Let $\cH_t$ be the following sub-hypergraph. $\cH_t$ consists of all hypernodes in $S_t$. For each hyperedge $e$ that contains a node in $S_t$, we add an edge $e'=e \cap S_t$ to $\cH_t$ if $e$ contains {\em none} of the following hypernodes: (1)  a hypernode in set $S'$ of color $j>i$, and (2) a node in set $S''$ of color $j<i$ that is already decided to be {\em not} in the MIS. We can construct $\cH_t$ quickly since each server (hypernode) can decide locally whether each client (hyperedge) adjacent to it satisfies the above property or not. 

Now we compute MIS of $\cH_t$ by simulating $\cA$ to solve the subgraph-MIS problem on $G_t$ where the subgraph we want to solve is the subgraph $G'_t$ of $G_t$ representing $\cH_t$. Note that since $G_t$ has diameter $O(\log n)$, $\cA$ will finish in $T(n)$ time if we simulate $\cA$ on only $G_t$. However, we will actually simulate $\cA$ on {\em all}  graphs $G_{t_1}, G_{t_2}, \ldots$ of color $i$ {\em simultaneously}. Since each edge is contained in $O(\log^3 n)$ such graphs, we can finish simulating $\cA$ on all graphs in $O(T(n)\log^3 n)$ time. 

After we finish simulating $\cA$ on $\cH_t$, we use the solution as a solution of MIS of the original graph $\cH$; that is, we say that a hypernode is in the MIS of $\cH$ if and only if it is in the MIS of $\cH_t$. We now prove the correctness. Let $M_t$ be the MIS of $\cH_t$. First, observe that any hypernode in $M_t$ can be added to the MIS solution of $\cH$ without violating the independent constraint since $\cH_t$ contains all hyperedges of $\cH$ except those that contain some hypernode of higher color (which is not yet added to the MIS of $\cH$) and hypernode of lower color that is already decided not to be in the MIS of $\cH$. Secondly, the fact that any hypernode $v$ in $S_t$ that is not in $M_t$ implies that there is a hyperedge $e'$ in $H_t$ that contains all hypernodes in $H_t$ except $v$. Let $e$ be a hyperedge in $\cH$ such that $e'\subseteq e$. Note that $e$ does not contain any hypernode in other set $S_{t'}$ of the same color as $S_t$. Also observe that every hypernode in $e\setminus S_t$ must be already decided to be in the MIS of $\cH$ (otherwise, we will not have $e'=e\cap S_t$ in $\cH_t$). Thus, every hypernode in $e'$ except $v$ is already in the MIS of $\cH$ as well; in other words, $v$ cannot be in the MIS of $\cH$. This completes the correctness of the algorithm. Thus, after we finish simulating $\cA$ on graphs of all colors, we obtain the MIS of $\cH$. Since we need $O(T(n)\log^3 n)$ time for each color, we need $O(T(n)\log^4 n)$ time in total. 
\end{proof}
}

\begin{lemma} \label{thm:local}
MIS can be solved in $O(\log^2 n)$ rounds in the LOCAL models (both vertex-centric and server-client representations). 
\end{lemma}
\onlyLong{\begin{proof}}
  \onlyShort{\begin{proof}[Proof Sketch]}
  Using \Cref{claim:decomposition}, we partition the hypernodes of the input network into subgraphs each of which have $O(\log n)$ diameter and no two subgraphs assigned the same colour share a hyper edge.
\onlyLong{Our algorithm proceeds in the same way as in the proof of \Cref{thm:decomposition congest}, except that there is no congestion in the LOCAL model when we simulate $\cA$ (as specified in Lemma \ref{thm:decomposition congest}) on all graphs of color $i$.} 
\onlyShort{Note that there is no congestion in the LOCAL model when we simulate $\cA$ (as specified in Lemma \ref{thm:decomposition congest}) on all graphs of color $i$.}
  Thus, we need $O(T(n))$ time per color instead of $O(T(n)\log^3 n)$. Moreover, we can solve the subgraph-MIS problem on a network of $O(\log n)$ diameter in $O(\log n)$ time by collecting the information about the subgraph to one node, locally compute the MIS on such node, and send the solution back to every node. Thus, $T(n)=O(\log n)$. 
%
%
It follows that we can solve MIS on networks of any diameter in $O(\log^2 n)$ time. 
\end{proof}

\onlyShort{\vspace{-0.35cm}}
\subsection{$O(\log^{(d+4)!+4} n)$ time in the CONGEST model assuming constant dimension $d$}\label{sec:constant dimension}
\onlyShort{\vspace{-0.2cm}}

Let $(\cH, \cH')$ be an instance of the subgraph-MIS problem such that the network $G$ representing $\cH$ has $O(\log n)$ diameter. We now show that we can solve this problem in $O(\log^{(d+4)!} n)$ time when  $\cH'$ has a constant dimension $d$, i.e. $|e|\leq d$ for every hyperedge $e$ in $\cH'$. By \Cref{thm:decomposition congest}, we will get a $O(\log^{(d+4)!+4} n)$-time algorithm for the MIS problem in the case of constant-dimensional hypergraphs (of any diameter) which works in both vertex-centric and server-client representations and even in the CONGEST model. This algorithm is also an important building block for the algorithm in the next section. 


Our algorithm simulates the PRAM algorithm of Beame and Luby \cite{BeameL90} which was proved by Kelsen \cite{Kelsen92} to finish in $O(\log^{(d+4)!} n)$ time when the input hypergraph has a constant dimension $d$ and this running time was recently extended to any $d\leq \frac{\log\log n}{4\log\log\log n}$ by Bercea~et~al.~\cite{aravind2}\footnote{The original running time of Kelsen \cite{Kelsen92} is in fact $O((\log n)^{f(d)})$ where $f(d)$ is defined as $f(2)=7$ and $f(i)=(i-1)\sum_{j=2}^{i-1}f(j)+7$ for $i>2$. The $O(\log^{(d+4)!} n)$ time (which is essentially the same as Kelsen's time) was shown in \cite{aravind2}. We will use the latter running time for simplicity. Also note that the result in this section holds for all $d\leq \frac{\log\log n}{4\log\log\log n}$ due to \cite{aravind2}.}. 
The crucial part in the simulation is to compute a number $\zeta(\cH')$ defined as follows. For $\emptyset \neq x \subseteq V(\cH')$ and an integer $j$  with $1\leq j\leq d-|x|$ we define: 
$N_j(x, \cH') = \{y\subseteq V(\cH') \mid x\cup y \in E(\cH') \wedge x\cap y = \emptyset\wedge |y|=j\},$
and
$d_j(x, \cH') = (|N_j(x, \cH')|)^{1/j}.$
Also, for $2\leq i\leq d$, let\footnote{A note on the notation: \cite{BeameL90,Kelsen92} use $\Delta$ to denote what we use $\zeta$ to denote here. We use a different notation since we use $\Delta$ for another purpose.} 
$\zeta_i(\cH') = \max \{d_{i-|x|}(x, \cH') \mid x\subseteq V(\cH') \wedge 0<|x|<i\}$
and
$\zeta(\cH') = \max\{\zeta_i(\cH') \mid 2\leq i\leq d\}.$
We now explain how to compute $\zeta(\cH')$ in $O(\log^{(d+4)!} n)$ time. First, note that we can assume that every node knows the list of members in each hyperedge that contains it: this information is already available in the vertex-centric representation; and in the server-client representation, every hyperedge can send this list to all nodes that it contains in $O(d)$ time in the CONGEST model. Every node $v$ can now compute, for every $i$, 
$\zeta_i(v, \cH') = \max \{d_{i-|x|}(x, \cH') \mid x\subseteq V(\cH') \wedge 0<|x|<i \wedge v\in x\}.$
This does not require any communication since for any $x$ such that $v\in x$, node $v$ already knows all hyperedges that contain $x$ (they must be hyperedges that contain $v$). Now, we compute $\zeta(\cH') = \max\{\zeta_i(v, \cH') \mid 2\leq i\leq d \wedge v\in V(\cH')\}$ by computing through the breadth-first search tree of the network representing $\cH$ (this is where we need the fact that the network has $O(\log n)$ diameter). 

Once we get $\zeta(\cH')$, the rest of the simulation is trivial; we refer to the full paper for details.
\onlyLong{
We provide some detail here for completeness. We mark each hypernode in $\cH'$ with probability $p=\frac{1}{2^{d+1}\zeta(\cH')}$. If a hyperedge has all of its nodes
marked, unmark all of its nodes. Remove the hypernodes that are still marked from $\cH'$ and add them to the independent set. We also remove these hypernodes from $\cH'$, thus reducing the size of some hyperedges in $\cH'$. In the remaining hypergraph do the following: eliminate any edges properly containing another edge; remove any hypernodes that form a 1-dimension edge (i.e. remove every hypernode $v$ such that there is a hyperedge $\{v\}$); finally, remove isolated vertices (i.e., those not contained in any edge) and add them to the independent set. Let $\cH'$ be the resulting hypergraph. Repeat this procedure until there is no hypernodes left. It is easy to see that all steps (before we repeat the procedure) takes $O(1)$ rounds.
Kelsen \cite{aravind2} and Bercea~et~al.~\cite{aravind2} showed that we have to repeat this procedure only $O(\log^{(d+4)!} n)$ time (in expectation and with high probability) when $d\leq \frac{\log\log n}{4\log\log\log n}$ (there is no guarantee for any other values of $d$); so, our simulation finishes in $O(\log^{(d+4)!} n)$ rounds. 
}

\onlyShort{\vspace{-0.2cm}}
\subsection{$\Delta^{\epsilon}\log^{(1/\epsilon)^{O(1/\epsilon)}} n$ and $\Delta^{o(1)}n^{o(1)}$ Time in the CONGEST model}\label{sec:MIS Delta epsilon}
\onlyShort{\vspace{-0.2cm}}

%
%

We rely on a modification of Tur\'an's theorem, which states that a (two-dimensional) graph of {\em low} average degree has a {\em large} independent set (see e.g. Alon and Spencer \cite{AlonS08book}). We show that this theorem also holds for high-dimensional hypergraphs, and show further that such a large independent set can be found w.h.p when the network diameter is $O(\log n)$. 

\begin{lemma}[A simple extension of Tur\'an's theorem]\label{thm:Turan} Let $d\geq 2$ and $\delta\geq 2$ be any integers. Let $\cH$ be any hypergraph such that every hyperedge in $\cH$ has dimension at least $d$, there are $n$ hypernodes, and the average hypernode degree is $\delta$. (Note that the diameter of the network representing $\cH$ can be arbitrary.) If every node knows $\delta$ and $d$, then we can find an independent set $M$ whose size in expectation is at least  $\frac{n}{\delta^{1/(d-1)}}(1-\frac{1}{d})$ in $O(1)$ time. 
%
\end{lemma}
\onlyLong{
\begin{proof}
We modify the proof of Theorem 3.2.1 in \cite[pp.29]{AlonS08book}. Let $p=(1/\delta)^{1/(d-1)}$ (note that $p<1$) and $S$ be a random set of hypernodes in $\cH$ defined by $Pr[v\in S]=p$ for every hypernode $v$. Let $X=|S|$, and let $Y$ be the number of hyperedges in $\cH$ contained in $S$ (i.e. hyperedge $e\in E(\cH)$ such that $e\subseteq S$). For each hyperedge $e$, let $Y_e$ be the indicator random variable for the event $e\subseteq S$; so, $Y=\sum_{e\in E(\cH)} Y_e$. Observe that for any hyperedge $e$, 
$E[Y_e] = p^{|e|} \leq p^d$ 
since $e$ contains at most $d$ hypernodes. So, $E[Y] = \sum_{e\in E(\cH)} E[Y_e] \leq \frac{n\delta}{d}p^d$ (the inequality is because the number of hyperedges in $\cH$ is at most  $\frac{n\delta}{d}$). 
Clearly, $E[X]=np$; so, 
$$E[X-Y] \geq np-\frac{n\delta}{d}p^d = n p (1-\frac{\delta}{d}p^{d-1}) = n(\frac{1}{\delta})^{\frac{1}{d-1}}(1-1/d)$$
where the last equality is because $p=(\frac{1}{\delta})^{\frac{1}{d-1}}$. 
Our algorithm will pick such a random set $S$. (Every node can decide whether it will be in $S$ locally.) Then it selects one vertex from each edge of $S$ and deletes it. (This can be done in $O(1)$ time.) This leaves a set $S^*$ with at least $n(\frac{1}{\delta})^{\frac{1}{d-1}}(1-\frac{1}{d})$ hypernodes in expectation. All edges having been destroyed, $S^*$ is an independent set.
%
%
%
\end{proof}
}

\paragraph{Algorithm.} We use the following algorithm to solve the subgraph-MIS problem on a sub-hypergraph $\cH'$ of $\cH$, assuming that the network representing $\cH$ has $O(\log n)$ diameter. Let $n'=|V(\cH')|$. Let $d$ be an arbitrarily large constant. Let $\cH'_d$ be the sub-hypergraph of $\cH'$ where $V(\cH'_d)=V(\cH')$ and we only keep hyperedges of dimension (i.e. size) at least $d$ in $\cH'_d$. (It is possible that $\cH'_d$ contains no edge.) We then find an independent set of expected size at least $\frac{n'}{\Delta^{1/(d-1)}}(1-1/d)$ in $\cH'_d$, denoted by $S$; this can be done in $O(1)$ time by \Cref{thm:Turan} (note that we use the fact that $\delta\leq \Delta$ here). Let $\cH'_S$ be the sub-hypergraph of $\cH'$ induced by nodes in $S$. 
Note that $\cH'_S$ does not contain any hyperedge in $\cH'_d$ and thus has dimension at most $d$, which is a constant. So, we can run the $O(\log^{(d+4)!} n)$-time algorithm from \Cref{sec:constant dimension} to find an MIS of $\cH'_S$. We let $M'_S$ be such a MIS of $\cH'_S$. 

Our intention is to use $M'_S$ as part of some MIS $M'$ of $\cH'$. Of course, any hypernode $v$ in $V(\cH'_S)\setminus M'_S$ cannot be in such $M'$ since  $M'\cup \{v\}$ will contain some hyperedge $e$ in $\cH'_S$ which is also a hyperedge in $\cH'$.  
It is thus left to find which hypernodes in $V(H')\setminus S$ should be added to $M'_S$ to construct an MIS $M'$ of $\cH'$. To do this, we use the following hypergraph. Let $\cH''$ be the sub-hypergraph of $\cH'$ such that $V(\cH'')=V(\cH')\setminus S$ and for every hyperedge $e\in E(\cH')$, we add a hyperedge $e\cap V(\cH'')$ to $\cH''$ if and only if $e \subseteq M'_S \cup V(\cH'')$; in other words, we keep edge $e$ that would be ``violated'' if we add every hypernode in $\cH''$ to $M'$. 
We now find an MIS $M''$ of $\cH''$ by recursively running the same algorithm with $\cH''$, instead of $\cH'$, as a subgraph of $\cH$. The correctness follows from the following claim\onlyShort{ (see the full paper for the proof).}\onlyLong{.}



\begin{claim}
$M'=M'_S\cup M''$ is a MIS of $\cH'$. 
\end{claim}
\onlyLong{
\begin{proof} First, we show that $M'$ is an independent set of $\cH'$. Assume for a contradiction that there is a hyperedge $e$ in $\cH'$ such that $e\subseteq M'$. This means that $e \subseteq M'_S \cup V(\cH'')$ since $M'_S\cup M''\subseteq M'_S \cup V(\cH'')$. It follows from the construction of $\cH''$ that  there is an edge $e'=e\cap V(\cH'')$ in $\cH''$. Note that $e\cap V(\cH'')\subseteq M''$; in other words $e'\subseteq M''$.  This, however, contradicts the fact that $M''$ is an MIS in $\cH''$. 

Now we show that $M'$ is maximal. Assume for a contradiction that there is a hypernode $v$ in $V(\cH')\setminus M'$ such that $M'\cup \{v\}$ is an independent set. If $v$ is in $S$, then $M'_S\cup \{v\}$ is an independent set in $\cH'_S$ (since it is a subset of $M'\cup \{v\}$), contradicting the fact that $M'_S$ is an MIS in $\cH'_S$. So, $v$ must be in $V(\cH'')$. This, however, implies that $M''\cup \{v\}$ is an independent set in $\cH''$ (again, since it is a subset of $M'\cup \{v\}$), contradicting the fact that $M''$ is an MIS in $\cH''$.
\end{proof}
}
We now analyze the running time of this algorithm. Recall that $E[|S|]\geq \frac{n'}{\delta^{(1/(d-1))}}(1-1/d)$. In other words, the expected value of $|V(\cH'')|\leq (1- \frac{c(d)}{\Delta^{1/(d-1)}}) |V(\cH')|$ where $c(d)=\frac{1}{2}(1-1/d)$ is a constant which is strictly less than one (recall that $d$ is a constant). It follows  that the expected number of recursion calls is $O(\Delta^{\frac{1}{d-1}})$. 
%
Since we need $O(\log^{(d+4)!} n)$ time to compute $M'_S$ and to construct $\cH''$, the total running time is  $O(\Delta^{\frac{1}{d-1}}\log^{(d+4)!} n)$. By \Cref{thm:decomposition congest}, we can compute MIS on any hypergraph $\cH$ (of any diameter) in 
\longOnly{$$O(\Delta^{\frac{1}{d-1}}\log^{(d+4)!+4} n)$$}
\shortOnly{$O(\Delta^{\frac{1}{d-1}}\log^{(d+4)!+4} n)$}
time. 
For any constant $\epsilon>0$, we set $d=1+1/\epsilon$ to get the claimed running time of \longOnly{
\begin{align}
O(\Delta^{\epsilon}\log^{(5+1/\epsilon)!+4} n) = \Delta^{\epsilon}\log^{(1/\epsilon)^{O(1/\epsilon)}} n.
\label{eq:Delta epsilon time}
\end{align}}
\shortOnly{$O(\Delta^{\epsilon}\log^{(5+1/\epsilon)!+4} n) = \Delta^{\epsilon}\log^{(1/\epsilon)^{O(1/\epsilon)}} n.$}
%
%
%
Moreover, by the recent result of Bercea~et~al. \cite{aravind2}, we can in fact set $d$ as large as $\frac{\log\log n}{4\log\log\log n}.$  
%
%
\longOnly{
In this case, note that for some constant $c'$, 
$$(d+4)! = d^{c'd} = e^{c'd\log d}= e^{c'\cdot\frac{\log\log n}{c\log\log\log n}\cdot \log\log\log n} = \log^{1/10} n$$ 
where the last equality holds when we set $c=10c'$.
Thus, 
$$\log^{(d+4)!} n = \log^{\log^{1/10} n} n  = 2^{(\log^{1/10} n)\log\log n}=n^{o(1)}.$$
The running time thus becomes $\Delta^{o(1)}n^{o(1)}.$
}
\shortOnly{If we set $d=\frac{\log\log n}{c\log\log\log n}$ for some large enough constant $c$, the term $\log^{(d+4)!} n$ can be bounded by $n^{o(1)}$ and thus the running time becomes $\Delta^{o(1)}n^{o(1)}$.}

\onlyShort{We obtain the $O(\sqrt{n})$ time by modifying the PRAM algorithm of Karp, Upfal, and Wigderson \cite[Section 4.1]{KarpUW88}. This algorithm can be found in the full version.}
\onlyLong{
\subsection{$O(\sqrt{n})$ Time in the CONGEST model}\label{sec:sqrt n algo}

We obtain the $O(\sqrt{n})$ time by modifying the PRAM algorithm of Karp, Upfal, and Wigderson \cite[Section 4.1]{KarpUW88}. (Note that we do not need the fact that the network diameter is $O(\log n)$ for this algorithm.) Their algorithm is as follows. Let $v_1, v_2, \ldots, v_n$ be a random permutation of hypernodes. The algorithm gradually adds a hypernode to the independent set one by one, starting from $v_1$. It stops at some hypernode $v_k$ when $v_k$ cannot be added to the independent set. Thus, $v_1, \ldots v_{k-1}$ are added to the independent set; the algorithm removes these hypernodes from the graph. It also removes {\em all} hypernodes that cannot be added to the independent set (i.e. any $v$ such that $\{v_1, \ldots, v_{k-1}, v\}$ contains some hyperedge) and all hyperedges that contain them. It repeats the same process to find a MIS of the remaining graph. It is easy to show (see \cite{KarpUW88} for detail) that the union of a MIS of the remaining graph and $\{v_1, \ldots, v_{k-1}\}$ is a MIS or the input graph. The key to proving the efficiency of this algorithm is the following. 

\begin{claim}[\cite{KarpUW88}]\label{claim:Karp et al}
The expected number of removed hypernodes ($v_1, \ldots, v_{k-1}$ and hypernodes that cannot be added to the independent set) in the above process is $\Omega(\sqrt{n})$. 
\end{claim}
It follows almost immediately that we have to repeat the process only $O(\sqrt{n})$ times in expectation (see \cite[Appendix]{KarpUW88} for detail).
We now show how to modify this algorithm to our setting. Every hypernode $v$ picks a random integer $r(v)$ between $1$ and $n^2$. It can be guaranteed that hypernodes pick different numbers with high probability. Then every hypernode $v$ marks itself to the independent set if for any hyperedge $e$ that contains $v$, $r(v)<\max_{u\in e} r(u)$, i.e., its number is not the maximum in any hyperedge. We add all marked hypernodes to the independent set, remove them from the graph, and eliminate hypernodes that cannot be added to the independent set (i.e. a hypernode $v$ marks itself as ``eliminated'' if there is a hyperedge $e$ such that $e\setminus \{v\}$ is a subset of marked hypernodes). We then repeat this process until there is no hypernode left. 

Using \Cref{claim:Karp et al}, we show that our algorithm has to repeat only $O(\sqrt{n})$ times, as follows. Consider an ordering $v_1, \ldots, v_n$ where $r(v_i)<r(v_{i+1})$. This is a random permutation. Let $k$ be such that $v_1, \ldots, v_k$ are added to the independent set by Karp et al.'s algorithm and $v_{k+1}, \ldots, v_n$ are not. Observe that for every $1\leq i\leq k$ and every hyperedge $e$ that contains $v_i$,  $r(v_i)<\max_{u\in e} r(u)$ (otherwise edge $e$ will be violated when we add $v_1, \ldots, v_k$ to the independent set). In other words, our algorithm will also add $v_1, \ldots, v_k$ to the independent set (but it may add other hypernodes as well). It follows  that our algorithm will eliminate every hypernode that is eliminated by Karp et al.'s algorithm. In other words, the set of hypernodes removed by our algorithm is a superset of the set of hypernodes removed by Karp et al.'s algorithm. Thus, by \Cref{claim:Karp et al}, the expected number of hypernodes removed in each iteration of our algorithm is $\Omega(\sqrt{n})$. By the same analysis as Karp et al., our algorithm will need only $O(\sqrt{n})$ iterations in expectation. Each iteration can be easily implemented in $O(1)$ rounds, so our algorithm takes $O(\sqrt{n})$ time in expectation. 

}
\onlyShort{\vspace{-0.2cm}}
\section{Applications of Hypergraph MIS algorithms to standard graph problems}
\onlyShort{\vspace{-0.2cm}}
\label{sec:applications}
In this section we show that our distributed hypergraph algorithms have direct applications in the standard graph setting.
\onlyLong{As a first application of our MIS algorithm, we show how to solve the restricted minimal dominated set (\rmds) problem in \Cref{sec:rmds}.
We will use this \rmds-algorithm to obtain a distributed algorithm for solving the balanced minimal dominating set (\bmds) problem, thereby resolving an open problem of \cite{balanced-minimal}.}

\onlyShort{\vspace{-0.1cm}}
\onlyLong{\subsection{Restricted Minimal Dominating Set (\rmds)} \label{sec:rmds}}
\onlyShort{\paragraph{Restricted Minimal Dominating Set (\rmds)} \label{sec:rmds}}

We are given a (standard) graph $G = (V,E)$ and a subset of nodes $R \subseteq V$, such
that  $R$ forms a dominating set in $G$ (i.e., every node $v \in V$ is either adjacent to $R$ or belongs to $R$).
We are required to find a {\em minimal} dominating set that is a subset of $R$ and dominates $V$.
Since a minimal vertex cover is the complement of a maximal independent set, we can leverage our \mis algorithm (cf.\ \Cref{sec:hyper}).
To this end, we show that the \rmds problem can be solved by finding a minimal hitting set (or minimal vertex cover) on a specific hypergraph $H$.
The server client representation of $H$ is determined by $G$ and $R$ as follows:
For every vertex in $V$ we add a client (i.e.\ hyperedge) and, for every vertex in $R$, we also add a server. 
Thus, for every vertex $u \in V$, we have a client $e_u$ and, if $u \in R$, we also have a server $s_u$.
We then connect a server $s_u$ to a client $e_v$, iff either $u$ and $v$ are adjacent in $G$, or $u=v$.
\onlyLong{\Cref{algo:rmds} contains the complete pseudo code of this construction.}
\onlyShort{See the full paper for the complete pseudo code of this construction.}
Note that we can simulate this server client network on the given graph with constant overhead in the CONGEST model.
We have the following result by virtue of \Cref{thm:mis}:

\begin{theorem} \label{thm:rmds}
  \rmds can be solved in expected time $\tilde O(\min\{\Delta^{\eps}, n^{o(1)}\})$  (for any const.\ $\eps > 0$) on graph $G$ in the CONGEST model and in time $O(\log^2 n)$ in the LOCAL model where $\Delta$ is the maximum degree of $G$.
\end{theorem}
\onlyLong{
\begin{algorithm}[t]
  \begin{algorithmic}[1]
\item[] Let $R$ be the set of restricted nodes (which are part of the MDS).
\item[] Simulate a server client network $H$. Every node (locally) adds vertices to the clients $C$ resp.\ servers $S$, and simulates the edges in $H$.
\FOR{every node $u$}
  \STATE Node $u$ adds a client $e_u$ to $C$.
  \IF{$u \in R$}
  \STATE Node $u$ adds a server $s_u$ to $S$, and an edge $(s_u,e_u)$ to $E(H)$.
  \ENDIF
\ENDFOR
\FOR{all nodes $u$, $v$ where $(u,v) \in E(G)$}
  \STATE If server $s_u$ exists in $H$, add edge $(s_u,e_u)$ to $H$.
\ENDFOR
\item[]
\STATE Find an MIS on $H$ and let $O_{MIS} \subseteq S$ be the servers that are in the output set.
\FOR{every node $u$ where $s_u$ exists}
  \STATE If $s_u \notin O_{MIS}$, then node $u$ adds itself to the \rmds.
\ENDFOR
\end{algorithmic} 
  \caption{An \rmds-algorithm: Finding a minimal dominating set on a graph $G$ that is a subset of a given dominating set $R$.}
  \label{algo:rmds}
\end{algorithm}
}

\onlyShort{\vspace{-0.2cm}}
\onlyLong{\subsection{Balanced Minimal Dominating Set} \label{sec:bmds}}
\onlyShort{\paragraph{Balanced Minimal Dominating Set} \label{sec:bmds}}
\onlyShort{\vspace{-0.2cm}}
We define the \emph{average degree} of a (standard) graph $G$,
 denoted by $\delta$, as the total degrees of its vertices (degree of a vertex is its degree in $G$)  divided by the number of vertices in $G$.
A {\em balanced minimal dominating set (BMDS)} (cf.\ \cite{balanced-minimal})
is a minimal dominating set  $D$ in $G$ that minimizes the ratio of
the average degree of $D$ to that of the graph itself (the average degree of the set of nodes $D$ is defined as the average degree of the subgraph induced by $D$).
\onlyLong{The \bmds problem is motivated by applications in fault-tolerance and load balancing (see \cite{balanced-minimal} and the references therein).
For example, in a typical  application, an MDS
can be used to form  clusters with low diameter, with the nodes in the MDS being the ``clusterheads'' \cite{moscibroda-survey}. Each clusterhead is  responsible for monitoring the nodes that are adjacent to it.
Having an MDS with low degree is useful
in a resource/energy-constrained setting since the number of
nodes monitored \emph{per} node in the MDS will be low (on average). This
can lead to  better load balancing, and consequently less resource or
energy consumption per node, which is crucial for ad hoc and sensor
networks, and help in extending the lifetime of such networks while also leading to better fault-tolerance. 
 For example,
in an $n$-node star graph, the above requirements imply that it is better
for  the leaf nodes to form the MDS rather than the central node alone.   In
fact,  the average degree of the MDS formed by the leaf  nodes --- which is
1 --- is within a constant factor of   the average degree of a star (which is
close to 2), whereas the average degree, $n-1$, of the MDS consisting of the
central node alone is much larger.}
A {\em centralized} polynomial time algorithm for computing a \bmds with (the best possible in general \footnote{That is,
there exists graphs with average degree $\delta$, where this bound is essentially the optimal.}) average degree $O(\frac{\delta \log \delta}{\log \log \delta})$ was given in \cite{balanced-minimal}. A distributed algorithm that gives the same bounds was left a key open problem.
We now present a distributed variant of this algorithm (cf.\ Algorithm~\ref{algo:bmds}) that uses our hypergraph \mis-algorithm as a subroutine.
\onlyLong{Note that since the \bmds problem is defined on standard graphs, we assume that \Cref{algo:bmds} executes on a standard synchronous network adhering to the CONGEST model of communication.
}
\begin{algorithm}[t]
  \onlyShort{\scriptsize}
  \begin{algorithmic}
\item[]
 \STATE Nodes compute the average network degree $\delta$.
 ~
  \STATE Every node $u$ of degree $ > 2\delta$ marks itself with probability $\frac{\log t}{t}$ where $t = \frac{2 \delta \log \delta}{\log \log \delta }$.

\STATE Every node of degree $\leq 2 \delta$ marks itself.

\STATE If a node $v$ is not marked, and none of the neighbors of $v$ are marked, then $v$ marks itself.
\STATE Let $\textsc{marked}$ be the set of nodes that are marked. 
Invoke the RMDS algorithm (cf.\ \Cref{sec:rmds}) on $G$ where the restricted set is given by \textsc{marked}.
\STATE Every node that is in the solution set of the RMDS algorithm remains in the final output set.
  \end{algorithmic}
  \caption{A distributed \bmds-algorithm.}
  \label{algo:bmds}
\end{algorithm}

\begin{theorem} \label{thm:bmds}
Let $\delta$ be the average degree of a graph $G$.
There is a CONGEST model algorithm that finds a \bmds with average degree $O(\frac{\delta \log \delta}{\log \log \delta })$ in expected $\tilde O(D + \min\{\Delta^\eps, n^{o(1)}\})$ rounds, where $D$ is the diameter, $\Delta$ is the maximum node degree of $G$, and $\eps > 0$ is
any constant.
\end{theorem}
\onlyLong{
\begin{proof}
Computing the average degree in Step~1 of \Cref{algo:bmds} can be done by first electing a leader, then building a BFS-tree rooted at the leader, and finally computing the average degree by convergecast.

It was shown in \cite{balanced-minimal} that marking the nodes according to \Cref{algo:bmds} yields an average degree of $O(\frac{\delta \log \delta}{\log \log \delta })$.
The runtime bound follows since the first part of the algorithm can be done in $O(D)$ rounds and the running time of the \rmds-algorithm (cf. \Cref{thm:rmds}).
\end{proof}
}

\onlyLong{\subsection{Minimal Connected Dominating Sets (\mcds)} \label{sec:mcds}}
\onlyShort{\paragraph{Minimal Connected Dominating Sets (\mcds)} \label{sec:mcds}}
\onlyShort{\vspace{-0.2cm}}

Given a graph $G$, the \mcds problem requires us to find a minimal dominating set $M$ that is connected in $G$.
We now describe our distributed algorithm for solving \mcds in the CONGEST model (\onlyLong{ see \Cref{algo:mcds} for the complete pseudo code}\onlyShort{see the full paper for the complete pseudo code}) and argue its correctness.
We first elect a node $u$ as the leader using a $O(D)$ time algorithm of \cite{KPPRT13:PODC}.
Node $u$ initiates the construction of a BFS tree $B$, which has $k\le D$ levels, after which every node knows its level (i.e.\ distance from the leader $u$) in the tree $B$.
Starting at the leaf nodes (at level $k$), we convergecast the maximum level to the root $u$, which then broadcasts the overall maximum tree level to all nodes in $B$ along the edges of $B$.

We then proceed in iterations processing two adjacent tree levels at a time, starting with nodes at the maximum level $k$.
Note that since every node knows $k$ and its own level, it knows after how many iterations it needs to become active.
Therefore, we assume for simplicity that all leafs of $B$ are on level $k$.
We now describe a single iteration concerning levels $i$ and $i-1$:
First, consider the set $L_i$ of level $i$ nodes that have already been added to the output set $M$ in some previous iteration;
initially, for $i=k$, set $L_i$ will be empty.
We run the $O(D + \sqrt{n})$ time algorithm of \cite{thurimella} to find maximal connected components among the nodes in $L_i$ in the graph $G$; let $\cC=\{C_1,\dots,C_\alpha\}$ be the set of these components and let $\ell_j$ be the designated component leader of component $C_j \in \cC$.


We now simulate a hypergraph that is defined as the following bipartite server client graph $H$:
Consider each component in $\cC$ as a \emph{super-node}; we call the other nodes on level $i$ \emph{non-super-nodes}.
The set $C$ of clients contains all super-nodes in $\cC$ and all nodes on level $i$ that are neither adjacent to any super-node nor have been added to the output set $O$.
The set $S$ of servers contains all nodes on level $i-1$.
The edges of $H$ are the induced inter-level edges of $G$ between servers and non-super-node clients. In addition, we add an edge between a server $s \in S$ and a (super-node) client $C_j \in \cC$, iff there exists a $v \in C_j$ such that $(v,s) \in E(G)$.
Conceptually, we can think of the edges incident to $C_j$ as pointing to the component leader node $\ell_j$.
Next, we find a MIS (cf.\ \Cref{sec:hyper}) on the (virtual) hypergraph $H$. \onlyShort{We refer to the full
paper for details.}
\onlyLong{
We sketch how we simulate the run of the MIS algorithm on $H$ in $G$:
If a node $v \in C_j$ receives a message from a node in $S$, then $v$ forwards this message to the component leader $\ell_j$. (If a node receives multiple messages at the same time, it simply forwards all messages sequentially by pipelining.) 
After waiting for $\tilde O(D)$ rounds, the component leader $\ell_j$ locally simulates the execution of $\ell_j$ according to the MIS algorithm by using the received (forwarded) messages.
Any messages produced by the simulation at $\ell_j$ are then sent back through the same paths to the neighbors of $C_j$.
Let $O_i$ be the set of nodes (on level $i-1$) that are not in the MIS; note that $O_i$ forms a minimal vertex cover on the hypergraph given by $H$.
At the end of this iteration, we add $O_i$ to the output set $M$ and then proceed to process levels $i-1$ and $i-2$.
}



\begin{theorem} \label{thm:mcds}
  \mcds can be solved in the CONGEST model in expected time $\tilde O(D (D\min\{\Delta^{o(1)}, n^{o(1)}\} +\sqrt{n}) )$.
\end{theorem}
\onlyLong{
\begin{proof}
We first argue the correctness of the algorithm.
It is straightforward to see that, after $k$ iterations, the \emph{solution set} $M= \bigcup_{i=1}^k O_i$ forms a dominating set of $G$.
For connectivity, note that since $O_i$ is a minimal vertex cover on the induced subgraph $H$, it follows that every super-node in the client set has a neighboring node in $O_i$.
This guarantees that $M$ remains connected (in $G$) after adding $O_i$.

Next, we consider minimality.
Suppose that there exists a node $w$ in the solution set $M$ that is \emph{redundant} in the sense that it can be removed from $M$ such that $M\setminus \{w\}$ is a MCDS of $G$.
Assume that $w$ became redundant in the iteration when processing levels $j$ and $j-1$.
Note that by the properties of the BFS tree, $w$ must be either on levels $j$ or $j-1$, since, in this iteration, we only add new nodes to $M$ that are themselves on level $j-1$.
By the correctness of the MIS algorithm, $w$ does not become redundant in the same iteration that it is added to $M$, thus $w$ can only be on level $j$.
Moreover, observing that $w$ can only have been added to $M$ in the preceding iteration 
to dominate some node $x$ on level $j+1$, it follows that $w$ cannot be made redundant by adding some node $z$ on level $j-1$, since $z$ cannot dominate $x$.
This shows that the set $M$ is minimal as required.

We now argue the running time bound.
The pre-processing steps of electing a leader and constructing a BFS tree can be completed in $O(D)$ rounds.
The for-loop of the MCDS algorithm has $O(D)$ iterations, thus it is sufficient if we can show that we can simulate a single iteration (including finding a MIS on the constructed hypergraph $H$) in $\tilde O(\sqrt{n} + D\min\{\Delta^{\eps},n^{o(1)}\})$ rounds.
Consider the iteration that determines the status of nodes in level $i$, i.e., the nodes on level $i$ form the set of servers as defined in MCDS algorithm.
First, we run the algorithm of \cite{thurimella}, which, given a graph $G$ and a subgraph $G'$, yields maximal connected components (w.r.t.\ $G'$) in time $\tilde O(D + \sqrt{n})$, where $D$ is the diameter of $G$.
Then, we simulate the MIS algorithm \Cref{sec:hyper} on the hypergraph $H$ given by the set of servers and the clients (some of which are super-nodes).
Consider a super-node $C_j$.
We can simulate a step of the MIS algorithm by forwarding all messages that nodes in $C_j$ receive (from servers on level $i-1$) to the component leader node $\ell_j$ by sequentially pipelining simultaneously received messages.
The following lemma shows that we can assume that each client has at most $O(\log n)$ incident servers, i.e., the dimension of the hypergraph $H$ is bounded by $O(\log n)$:
\begin{lemma} \label{lem:logn}
If there is an algorithm $\cA$ that solves MIS on $n$-node $m$-edge hypergraphs of dimension up to $3\log (m+n)$ in $T(n)$ rounds for some function $T$, then there is an algorithm $\cA'$ that solves hypergraph MIS on any $n$-node $m$-edge hypergraph with any dimension  in $\tilde O(T(n))$ rounds. 
\end{lemma}
\onlyLong{
\begin{proof}[Proof of \Cref{lem:logn}]
$\cA'$ works as follows. Let $\cH$ be the input graph. We will use $M$ as a final MIS solution for $\cA'$; initially, $M=\emptyset$. First, we mark every hypernode with probability $1/2$. Let $\cH'$ be the subgraph of $\cH$ induced by marked nodes (i.e. $\cH'$ consist of every edge such that every node it contains is marked). Observe that, with probability at least $1-1/m^2$, every hyperedge in $\cH'$ has dimension at most $3\log m$ because every hyperedge that contains more than $3\log m$ nodes will have all its nodes marked with probability at most $m^3$. We now run $\cA$ to solve hypergraph MIS on $\cH'$. We add all nodes in the resulting MIS to $M$ and remove them from $\cH$. Additionally, we remove from $\cH$ all other nodes in $\cH'$ (that are not in the MIS of $\cH'$) and edges containing them. (These nodes cannot be added to the MIS of $\cH$ so they are removed.) We then repeat this procedure to find the MIS of the remaining graph. Observe that this procedure removes $n/2$ nodes in expectation. So, we have to repeat it only $O(\log n)$ times in expectation. Each of this procedure takes $O(T(n))$ time, so we have the running time of $\tilde O(T(n))$ in total. 
\end{proof}
}
It follows from \Cref{lem:logn} that forwarding messages towards the component leader can incur a delay of at most $O(\log n)$ additional rounds due to congestion.
This means that one step of the MIS algorithm can be implemented in $\tilde O(D)$ rounds, and thus the total time complexity of a single iteration of the for-loop takes time $\tilde O(D\min\{\Delta^{\eps}, n^{o(1)}\} + \sqrt{n})$, as required.
\end{proof}
}

\onlyLong{
\begin{algorithm}[t]
  \begin{algorithmic}[1]
\item[]
  \STATE Let $M$ be the final output set; initially $M=\emptyset$.
  \STATE We perform leader election using an $O(D)$ time algorithm of \cite{KPPRT13:PODC}, yielding some leader $\ell$.
  \STATE Node $\ell$ initiates the construction of a breadth-first-search tree $B$ of  $k\le D$ levels.
  \STATE The leafs of $B$ report their level (i.e.\ distance from the root) to $\ell$ by convergecast and the leader $\ell$ then rebroadcasts the maximum level to all children along the tree edges.
  At the end of this step, every node knows its level  in $B$ and the maximum tree level.
  \item[]
  \FOR{tree level $i=k,\dots,1$}
  \STATE Let $L_i \subseteq M$ denote the nodes on level $i$ that have been added to $M$. (Note that $L_k$ is empty initially.)
  Find a set of maximal connected components $\cC=\{C_1,\dots, C_\alpha\}$ of the nodes in $L_i$ using the $O(D+\sqrt{n})$ time algorithm of \cite{thurimella}; let $\ell_1,\dots,\ell_\alpha$ denote the roots of the respective components.
\item[]
  \item[] \textsl{Solving MIS on the hypergraph induced by levels $i$ and $i-1$:}
  \STATE We construct the following bipartite server client graph $H$.
  Consider each component in $\cC$ as a ``super-node''.
  The set $C$ of clients contains all super-nodes in $\cC$ and all nodes on level $i$ that are neither adjacent to any super-node nor have been added to the output set $O$.
  The set $S$ of servers contains all nodes on level $i-1$.
  The edges of $H$ are the induced inter-level edges of $G$ between servers and clients that do not form a component. In addition, add an edge between $s \in S$ and $C_j \in C$, iff there exists an $v \in C_j$ such that $(v,s) \in E(G)$.
  Conceptually, we can think of the edges incident to $C_j$ to point to the component leader node $\ell_j$.
  \STATE Find a MIS (cf.\ \Cref{sec:hyper}) on the virtual hypergraph $H$:
 \item[] We sketch how we simulate the run of the MIS algorithm on $H$ in $G$:
  If a node $v \in C_j$ receives a message from a node in $S$, $v$ forwards this message to the component leader $\ell_j$. (If a node receives multiple messages at the same time, it simply forwards all messages sequentially by pipelining.) 
  After waiting for $\tilde O(D)$ rounds, the component leader $\ell_j$ locally simulates the execution of the MIS algorithm by using the received (forwarded) messages.
  Any messages produced by the simulation at $\ell_j$ are then sent back through the same paths to the neighbors of $C_j$.
  \STATE Add every node on level $i-1$ that is not in the MIS to the output set $M$.
  \ENDFOR
\item[]
  \end{algorithmic}
  \caption{A distributed \mcds-algorithm.}
  \label{algo:mcds}
\end{algorithm}
}


\onlyLong{
\vspace{-0.5cm}
\subsection{Lower Bounds} \label{sec:lb}
\vspace{-0.1cm}
\onlyLong{In this section}\onlyShort{In the full paper} we show lower bounds on the number of rounds for computing a minimal connected dominating set (\mcds).
First, we show that $\tilde{\Omega}(D + \sqrt{n})$ rounds are necessary in the worst case for solving \mcds in the CONGEST model showing a reduction to the spanning connected subgraph problem (\scs).  

We then consider the LOCAL model where nodes can send messages of unbounded size.
Here we present a lower bound
of $\Omega(D)$ rounds for computing a minimal connected dominating set (\mcds).
While it is easy to see that this lower bound holds on a cycle of $n$ nodes, we
show that $\Omega(D)$ is a \emph{universal} bound in the sense that, for any given diameter $D=D(n)$ as a function of $n$, we can construct a graph where the algorithm takes $\Omega(D)$ time.
As a byproduct of our proof, we obtain the same lower bound for spanning tree computation and maximal clique.
}
\onlyLong{
\begin{theorem} \label{thm:CONGEST-lb}
There exists an $\epsilon > 0$ and a graph $G$ of $n$ nodes and diameter $D$, such that any $\epsilon$-error \mcds algorithm $R$ takes $\tilde{\Omega}(D + \sqrt{n})$ rounds in the CONGEST model.
\end{theorem}
\begin{proof}
We will show the lower bound via reduction from the spanning connected subgraph (\scs) problem, for which an $\tilde{\Omega}(D + \sqrt{n})$ lower bound is already known (cf.\ \cite{STOC11}).
Consider an instance of the \scs problem: we are given a set of edges defining a subgraph $H$ of graph $G$ and every node must output ``yes'' if $H$ spans $G$ and is connected; otherwise at least one node must output ``no''.

Suppose that we are given the \mcds algorithm $R$ as stated in the theorem.
We will first show that, as long as graph $G$ has $O(n)$ edges, we can instantiate $R$ to yield a solution for \scs without significant overhead.
Since the lower bound graph for the \scs problem in \cite{STOC11} has $O(n)$ edges, this will yield the result.
}
\onlyLong{
For an instance of \scs given by $G$ and $H \subseteq G$, we will construct a graph $G' = G'(G,H)$ of $\Theta(|V(G)|)$ vertices and $\Theta(|E(G)|)$ edges:
We initialize $V(G')$ to $V(G)$ and subdivide each edge $(u,v) \in E(G)$ by adding a \emph{subdividing vertex} $b_{u,v}$ to $V(G')$ and
edges $(u,b_{u,v})$, $(b_{u,v},v)$ to $E(G')$.
Let $B[H]$ be the set of all vertices that subdivide an edge in $H$ and let $B[G\setminus H]$ be the set of vertices subdividing other edges in $E(G) \setminus E(H)$.
Then, for every vertex $u \in V(G) \cup B[H]$ we add an \emph{outer vertex} $g_u$ to $V(G')$ and attach it to $u \in V(G')$ by adding the edge $(g_u,u)$ to $E(G')$.
In other words, we attach outer vertices to all nodes that were part of the original graph $G$ and to all nodes that subdivide an edge in $H$.

\begin{lemma} \label{lem:scs}
If $M$ is an minimal connected dominating set of $G'$, then the following holds:
$\forall u \in M\colon u \notin B[G\setminus H]$ if and only if $H$ is a spanning connected subgraph of $G$'
\end{lemma}
\onlyLong{
\begin{proof}
First, observe that to dominate an outer vertex $g_u$, it is necessary that either $g_u \in M$ or its (only) neighbor $v$ is in $M$.
In the former case, it follows that $v$ must also be in $M$ to satisfy connectivity.
But then we could remove $g_u$ from $M$ and still guarantee domination; thus it follows that no outer vertex is in $M$ and, every neighbor of an outer vertex is in $M$.
In particular, this means that all vertices of $G$ and all vertices that subdivide edges of $H$ must be in $M$ (since for each of these we added an outer vertex).
Finally, we observe that $H$ is not a connected subgraph if and only if $M$ needs to contain vertices that subdivide edges \emph{not} in $H$.
This completes the proof of \Cref{lem:scs}.
\end{proof}
}
Armed with \Cref{lem:scs}, we can simply invoke algorithm $R$ to test whether $H$ is a spanning connected subgraph of $G$.
Assuming that $|E(G)| \in O(n)$, it follows that asymptotically $G'$ and $G$ have the same number of vertices and edges and thus it is straightforward to simulate the run of the MCDS algorithm $R$ on the (virtual) graph $G'$ on top of the actual network $G$.
From \cite{STOC11} we know that there exists a graph of $n$ nodes and $O(n)$ edges where \scs takes time $\Omega(D + \sqrt{n})$.
This completes the proof of \Cref{thm:CONGEST-lb}.
\end{proof}
}
\onlyLong{
\begin{theorem}[Universal Lower Bound] \label{thm:MCDS-ST-lowerbound}
Let $R$ be an algorithm that computes a Spanning Tree  (resp.\ maximal clique and \mcds) in the LOCAL model with probability at least $15/16+\eps$, for any constant $\eps>0$.
Then, for every sufficiently large $n$ and every function $D(n)$ with $2 \le D(n) < n/4$, there exists a graph $G$ of $n' \in \Theta(n)$ nodes and diameter $D' \in \Theta(D(n))$ where $R$ takes $\Omega(D)$ rounds with constant probability.
\end{theorem}

\begin{proof}
For a given $n$ and a function $D(n)$, we construct the following lower bound graph $G$ of $\Theta(n)$ nodes and diameter $\Theta(D(n))$.
Let $d \ge 1$ be the largest integer such that $n = 4d D(n) + \ell$, for $0 \le \ell < 4D$; we will construct a graph $G$ of $n' = n - \ell \in \Theta(n)$ vertices and diameter $D' = \lfloor (n - \ell)/8d \rfloor \in \Theta(D)$.
Let $u_0,\dots,u_{n'-1}$ be the vertices of $G$.
We will consider the set of \emph{bridge vertices} defined as $\{u_i \in V(G) \mid \exists k\ge 0\colon i = k d < n'\}$ to describe the edges of $G$; vertices not in this set are the \emph{non-bridge vertices} of $G$.
That is, for every bridge vertex $u_i$, we add the arc edges $(u_i,u_{i+1}),(u_i,u_{i-1}),\dots,(u_i,u_{i+d}),(u_i,u_{i-d})$ (indices are modulo $n'$).
See \Cref{fig:lbGraph} for a concrete instance of this graph.
{
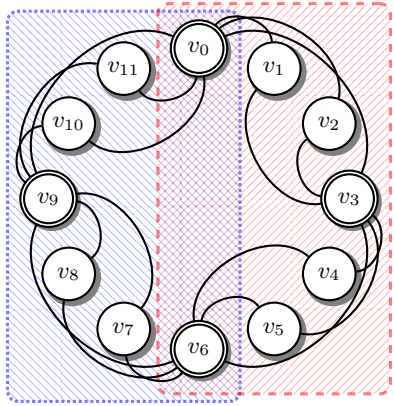
\begin{figure}[t]
\begin{center}
\pgfdeclarelayer{background}
\pgfdeclarelayer{inbetween}
\pgfsetlayers{background,inbetween,main}
\begin{tikzpicture}
\tikzstyle{link}=[-,black,thick,auto,bend angle=70,bend left]
\tikzstyle{arc}=[draw=black,rectangle,rounded corners,fill=lightgray,font=\normalsize]
\tikzstyle{arcNeighborhoodZ}=[semitransparent,very thick,draw=black,densely dotted,rectangle,rounded corners]
\tikzstyle{arcNeighborhoodZ1}=[arcNeighborhoodZ,dashed]
\small 
\node[b]  (v0)  at (90:2cm)                      {$v_{0}$};
\node[v]  (v1)  at (90-360/12:2cm)                {$v_{1}$};
\node[v]  (v2)  at (90-2*360/12:2cm)               {$v_{2}$};

\node[b]  (v3)  at (90-3*360/12:2cm)               {$v_{3}$};
\node[v]  (v4)  at (90-4*360/12:2cm)               {$v_{4}$};
\node[v]  (v5)  at (90-5*360/12:2cm)               {$v_{5}$};

\node[b]  (v6)  at (90-6*360/12:2cm)               {$v_{6}$};
\node[v]  (v7)  at (90-7*360/12:2cm)                 {$v_{7}$};
\node[v]  (v8) at (90-8*360/12:2cm)                {$v_{8}$};

\node[b]  (v9) at (90-9*360/12:2cm)                {$v_{9}$};
\node[v]  (v10) at (90-10*360/12:2cm)               {$v_{10}$};
\node[v]  (v11) at (90-11*360/12:2cm)                {$v_{11}$};

\draw[link] (v0) to (v1);
\draw[link,bend left] (v0.north east) to (v2);
\draw[link] (v0) to (v3);
\draw[link] (v0) to (v11);
\draw[link] (v0) to (v10);

\draw[link] (v3) to (v4);
\draw[link] (v3) to (v5);
\draw[link] (v3) to (v6);
\draw[link] (v3) to (v2);
\draw[link] (v3) to (v1);

\draw[link] (v6) to (v7);
\draw[link] (v6) to (v8);
\draw[link] (v6) to (v9);
\draw[link] (v6) to (v5);
\draw[link] (v6) to (v4);

\draw[link] (v9) to (v10);
\draw[link] (v9) to (v11);
\draw[link] (v9) to (v0);
\draw[link] (v9) to (v8);
\draw[link] (v9) to (v7);

\begin{pgfonlayer}{background}
  \node[arcNeighborhoodZ,fit=(v9) (v10) (v11) (v8) (v7) (v6) (v0),label={[yshift=1.30cm,font=\normalsize]below:$$},yshift=-0.1cm,scale=1.05,blue,fill=blue!40,pattern=north west lines,pattern color=blue] (C0) {};
  \node[arcNeighborhoodZ1,fit=(v3) (v0) (v1) (v2) (v4) (v5) (v6),label={[xshift=1.35cm,font=\normalsize]left:$$},scale=1.05,red,fill=red!70,pattern=north east lines,pattern color=red] (C1) {};
\end{pgfonlayer}
\end{tikzpicture}
\end{center}
\caption{The Lower Bound Graph of Theorem~\ref{thm:MCDS-ST-lowerbound} for $n=12$ and diameter $2$. Bridge vertices are marked by a double frame. The shaded regions represent two bridge clusters that partition the vertices into edge-disjoint sets.}
\label{fig:lbGraph}
\end{figure}
}

We first observe that solving maximal clique in this graph provides a leader node by simply running the $O(1)$ time leader election algorithm of \cite{KPPRT-LE-ICDCN13} on the clique, thus showing that maximal clique takes $\Omega(D)$ time. 

We now describe how to solve the leader election problem on $G$ given an \mcds-algorithm or an \st-algorithm. 
Let $b_0,\dots,b_{n'/4d-1}$ be an ordering of the bridge vertices according to their adjacencies in $G$.
As there are no edges between non-bridge vertices, any \mcds $M$ must contain all except possibly $1$ bridge vertex to guarantee connectivity.
Moreover, the fact that every bridge vertex $b_i$ dominates $b_{i-1}$ and $b_{i+1}$ (modulo $n'/4d$) implies that $M$ must omit a bridge vertex to be minimal.

\begin{observation} \label{obs:mcds}
If $M$ is an \mcds of $G$, then there is exactly one bridge vertex $b_i \in G$ such that $b_i \notin M$.
\end{observation}

We call the subgraph that consists of a bridge vertex $b_i$ and its adjacent vertices a \emph{bridge cluster}.
Analogously to Observation~\ref{obs:mcds}, we have the following:

\begin{observation} \label{obs:st}
Let $B$ be a partitioning of $G$ into edge-disjoint bridge clusters and let $S$ be a spanning tree of $G$.
Then, there is exactly one bridge cluster $b \in B$ such that the subgraph $b \cap S$ is disconnected.
\end{observation}

Suppose that $R$ is an algorithm that solves \st (the argument is analogous for \mcds) with probability $p$ in time $T$.
We first run $R$ to obtain a spanning tree of $G$ and then instruct every bridge vertex to check whether its cluster is connected.
By construction, every vertex locally knows if it is a bridge vertex since non-bridge vertices have exactly $2$ edges while bridge vertices have degree $>2$.
By Observation~\ref{obs:st}, exactly $2$ bridge vertices $b_i$ and $b_{i+1}$ will determine that their (overlapping) clusters are disconnected.
The nodes $b_i$ and $b_{i+1}$ determine which of them has the greater id; this node then elects itself as the leader, while all other nodes enter the non-elected state.
Thus there is an algorithm that elects a leader in $O(T)$ rounds with probability $p$.

It was shown in Theorem~3.13 of \cite{KPPRT13:PODC} that  there is a class of graphs $G_n$ with diameter $D(n)$ such that leader election takes $\Omega(D(n))$ rounds with constant probability.
The proof of this result relies on the fact that the vertices of $G_n$ can be partitioned into $4$ disjoint but symmetric sets $C_1,\dots,C_4$ such that the distance between $C_1$ and $C_3$ (resp.\ $C_2$ and $C_4$) is $\Omega(D)$.
It is straightforward to check that these properties also hold true in our graph class $G$.
In particular, all bridge vertices observe the same round $r$-neighborhood of $G$, for all $r \ge 1$.
Thus the proof of Theorem~3.13 in \cite{KPPRT13:PODC} can be adapted to our graph $G$.
(We defer the details of this adaptation to the full version of the paper.)
Together with the above reduction from leader election, 
this implies the sought time bound of $\Omega(D)$ rounds (with constant probability) for computing a minimal connected dominating set and finding a spanning tree on $G$.
\end{proof}
}
%
%

\onlyLong{
\vspace{-0.5cm}
\section{Distributed Algorithms for Other Hypergraph Problems}
\vspace{-0.2cm}
\label{sec:other}
Many algorithms in this section will simulate an algorithm for finding  a MIS on a (standard) graph developed by Luby \cite{Luby86} as a subroutine. 
One version of this algorithm is this: (1) Randomly assign unique priorities to nodes in $G$ (which can be achieved w.h.p. by having each node in $G$ randomly pick an integer between $1$ and $n^4$). (2) We mark and add all nodes that has higher priority than all its neighbors to the independent set. (3) We remove these marked nodes and their neighbors from the graph and repeat the procedure. Luby \cite{Luby86} shows that this procedure will repeat only $O(\log n)$ times in expectation. So, it is sufficient to get $\tilde O(1)$ time if our algorithms can simulate the three steps above in $\tilde O(1)$ time.

\onlyLong{\subsection{Maximal Clique} \label{sec:clique}}


\begin{theorem}
Maximal clique can be computed in $\tilde O(D)$ time in the CONGEST vertex-centric model and $\tilde O(D+\dimension)$-time in the CONGEST server-client model, where $D$ is the network (i.e., server graph or the server-client bipartite graph) diameter and $\dimension$ is the hypergraph dimension.
\end{theorem}
\onlyLong{
\begin{proof}
Recall that in this problem, we want a maximal set $S$ of hypernodes such that every two hypernodes $u$ and $v$ in $S$ are contained in some common hyperedge. This is equivalent to finding a maximal clique in the server graph (defined in \Cref{sec:prelim}). 

Since the underlying network of the vertex-centric model is exactly the server graph, we can easily find a maximal clique in this model, as follows. Pick any node $s$. (This can be done in $O(D)$ time by, e.g. picking a node with smallest ID or using a leader election algorithm.) Let $S$ be the set of all neighbors of $s$. Let $G_S$ be the subgraph of the server graph induced by nodes in $S$. Observe that if $M$ is a maximal clique in $G_S$ then $\{s\}\cup M$ is a maximal clique in $G$. So, it is sufficient to find a maximal clique in $G_S$. Observe further that if $\bar{G}_S$ is the complement graph of $G_S$ (i.e. an edge $(u, v)$ is in $\bar{G}_S$ if and only if it is not in $G_S$), then finding a maximal clique in $G_S$ is equivalent to finding a MIS in $\bar{G}_S$. 

We now simulate Luby's algorithm to find a MIS in $\bar{G}_S$.  We simulate the first step by letting node $s$ generate a random permutation of nodes in $\bar{G}_S$, say $v_1, v_2, \ldots, v_{|S|}$, and send a priority $i$ to node $v_i$. This can be done in one round since all nodes in $\bar{G}_S$ are neighbors of $s$. Now, for every node $v$ in $\bar{G}_S$ of priority, say $i$,  checks whether its priority is higher than all its neighbors in $\bar{G}_S$ (as required by the second step of Luby's algorithm). Observe that this is the case if and only if the priorities $i+1, i+2, \ldots |S|$ are given to $v$'s neighbors in $G_S$. Node $v$ can check this in one round by receiving the priorities of all its neighbors in $G_S$. For simulating the third step, each node $v$ has to know whether it has a neighbor in $\bar{G}_S$ that is marked. We do this by counting the number of marked nodes (every node tells $s$ whether it is marked or not). Let $c$ be such number. Then, every node $v$ counts how many of its neighbors in $G_S$ are marked. If this is less than $c$, then $v$ has a neighbor in $\bar{G}_s$ that is marked. This takes $O(1)$ rounds.

The above simulation of Luby's algorithm can be extended to the server-client model with an extra $O(\dimension)$ factor cost: For $s$ to distribute the priorities in the first step, it has to send up to $\dimension$ priorities to the same hyperedge. For the second step, where each node $v$ has to check whether its priority is higher than all its neighbors in $\bar{G}_S$, $v$ has to receive the priorities of all its neighbors in $G_S$, and it might have to receive up to $\dimension$ priorities from the same hyperedge. Finally, for the third step where every node has to know the number of neighbors in $G_S$ that are marked, it has to received the list of IDs of its marked neighbors, and it might have to receive up to $\dimension$ IDs from the same hyperedge. 
\end{proof}
}
Note that the dependence on the diameter in the running time is necessary, as shown in \Cref{thm:MCDS-ST-lowerbound}.
\onlyLong{\subsection{$(\Delta+1)$-Coloring}}
%
\onlyShort{
  The $(\Delta+1)$-coloring problem requires a coloring of the nodes such that no hyperedge is monochromatic. This can achieved by ensuring that at least vertices per hyperedge have distinct colors, which can be ensured using standard graph coloring (cf.\ full paper).
Also, by solving MIS on the line graph of a given hypergraph, we can get a maximal matching.
\begin{theorem}
%
The $(\Delta+1)$-coloring problem on hypergraphs has the same complexity as the $(\Delta+1)$-coloring problem on standard (two-dimensional) graphs; in particular, it can be solved in $O(\log n)$ time. This holds in both vertex-centric and server-client representations and in the CONGEST model. 
The maximal matching problem on hypergraphs can be solved in $O(\log n)$ time in the CONGEST server-client model. 
\end{theorem}
 }
 \onlyLong{
\begin{theorem}
%
The $(\Delta+1)$-coloring problem on hypergraphs has the same complexity as the $(\Delta+1)$-coloring problem on standard (two-dimensional) graphs; in particular, it can be solved in $O(\log n)$ time. This holds in both vertex-centric and server-client representations and even in the CONGEST model. 
\end{theorem}

\begin{proof}
Recall that in the $(\Delta+1)$-coloring problem we want to color hypernodes so that there is no monochromatic hyperedge, i.e. all hypernodes it contains have the same color. We solve this problem by converting a hypergraph $\cH$ to a two-dimensional graph $G$ on the same set of nodes as follows. For every hyperedge $e$ in $\cH$, pick arbitrary two distinct hypernodes it contains, say $u$ and $v$, and create an edge $e'=(u, v)$ in $G$. Observe that $G$ has maximum degree at most $\Delta$ and any valid coloring in $G$ will be a valid coloring in $H$ (since if an edge $e$ in $\cH$ is monochromatic, then the corresponding edge $e'$ in $G$ will also be monochromatic). Thus, it is sufficient to find a $(\Delta+1)$ coloring in $G$. We can do this by simulating any $(\Delta+1)$-coloring algorithm for $G$ on $\cH$. This shows that $(\Delta+1)$-coloring on hypergraphs is {\em as easy as}  $(\Delta+1)$-coloring on standard graphs. 
\end{proof}
}


\onlyLong{\subsection{Maximal Matching}
\begin{theorem}
The maximal matching problem on hypergraphs can be solved in $O(\log n)$ time in the CONGEST server-client model. 
\end{theorem}
}
\onlyLong{
\begin{proof}
Recall that this problem on a hypergraph $\cH$ asks for a maximal set $S\subseteq E(\cH)$ of {\em disjoint} hyperedges, i.e. $e\cap e'=\emptyset$ for all $e\neq e'$ in $S$. Consider the following {\em line graph} $G$: nodes of $G$ is the hyperedges in $\cH$, i.e. $V(G)=E(\cH)$, and there is an edge between two nodes $e, e'\in V(G)$ if and only if their corresponding hyperedges overlap, i.e. $e\cap e'\neq \emptyset$. Clearly, a set $S$ is a maximal matching in $\cH$ if and only if it is a MIS in $G$. Thus, it is left to find a MIS in $G$. 
This can be done by simulating Luby's algorithm \cite{Luby86}.
Observe that the first and second steps need no communication. For the third step, every node in $G$ (hyperedges in $\cH$) only needs to know the highest priority among its neighbors. This can be done in $O(1)$ rounds by having each hyperedge (client in the server-client representation) in $\cH$ send its priority to all hypernodes (server) that it contains, then these hypernodes sends the maximum priority that it receives to all hyperedges that contain it. 
So, we can implement the three steps of Luby's algorithm in $O(1)$ rounds.
\end{proof}
}
}

\section{Concluding Remarks and Open Problems}
\label{sec:conc}
\onlyShort{\vspace{-0.1in}}

Our work shows that while some local symmetry breaking problems such as coloring and maximal matching can be solved in polylogarithmic rounds in both the LOCAL and CONGEST models, for many others such as MIS, hitting set, and maximal clique it remains a challenge to obtain polylogarithmic time algorithms in the CONGEST model. This dichotomy manifests in hypergraphs
of higher dimension. Understanding this dichotomy can be helpful to make further progress 
in improving the bounds or showing lower bounds, especially in the CONGEST model. In particular, an important open question is whether we can show super-polylogarithmic lower bounds
for MIS for hypergraphs of high dimension in the CONGEST model?

Our results also have implications to solving hypergraph problems in the classical PRAM model.
Our CONGEST model algorithms can be translated into PRAM algorithms running in (essentially) the same number of rounds (up to polylogarithmic factors).
 In particular,  improving over the $\tilde O(\Delta^{o(1)})$ round algorithm for MIS in the CONGEST model can point to better PRAM algorithms for MIS which has been eluding researchers till now. A major question is whether $O(\polylog n)$ or even $O(\polylog m)$  round algorithms are possible
 in the CONGEST model for MIS (as shown here, the answer is ``yes'' in the LOCAL model).

Another aspect of this work, which was one of our main motivations, is using hypergraph algorithms for  solving problems in graphs efficiently.  In particular, our hypergraph MIS algorithm leads to  fast distributed algorithms for the  BMDS and the MCDS problems.
In particular, it will be interesting to see if one can give an algorithm for MCDS that essentially matches
the lower bound of $\tilde{\Omega}(D + \sqrt{n})$ (when $D$ is large).
\enlargethispage{1\baselineskip}


\onlyShort{\vspace{-0.2in}}
\bibliography{papers}

\end{document}